\documentclass[a4paper,11pt]{article}

\usepackage{geometry}
\usepackage[main=british, spanish]{babel}
\usepackage[utf8]{inputenc}

\usepackage{multirow} 
\usepackage{longtable} 
\usepackage{tabularray} 
\UseTblrLibrary{diagbox}
\usepackage{amsmath}
\usepackage{amssymb}
\usepackage{amsthm}
\usepackage{amsfonts}
\usepackage{enumerate}

\allowdisplaybreaks

\usepackage[normalem]{ulem}

\usepackage[shortlabels]{enumitem}
\usepackage[hidelinks]{hyperref} 

\usepackage[capitalise, nameinlink]{cleveref}
\crefname{subsection}{Subsection}{Subsections} 

\newtheoremstyle{estiloteorema}
  {\topsep}
  {\topsep}
  {\em}
  {}
  {\bfseries}
  {.}
  { }
  {\thmname{#1}\thmnumber{ #2}\thmnote{ (#3)}}  

\theoremstyle{estiloteorema}
\newtheorem{theorem}{Theorem}[section]
\newtheorem{proposition}[theorem]{Proposition}
\newtheorem{lemma}[theorem]{Lemma}

\newtheorem{remark}[theorem]{Remark}

\theoremstyle{definition}
\newtheorem{notation}[theorem]{Notation}

\newtheorem{example}[theorem]{Example}

\usepackage{marvosym}

\newlist{propenum}{enumerate}{1}
\setlist[propenum]{label=\roman*), ref=\theproposition(\roman*)}
\crefalias{propenumi}{proposition}

\newlist{corolenum}{enumerate}{1}
\setlist[corolenum]{label=\emph{(}\roman*\emph{)}, ref=\thecorollary(\roman*)}
\crefalias{corolenumi}{corollary}

\usepackage[colorinlistoftodos,textwidth=20mm]{todonotes}
\newcommand{\F}{\mathbb{F}}

\newcommand{\Fq}{\mathbb{F}_q}

\newcommand{\Fqr}[1]{\mathbb{F}_{q^{#1}}}
\DeclareMathOperator{\lcm}{lcm}
\DeclareMathOperator{\gal}{Gal}
\DeclareMathOperator{\End}{End}
\DeclareMathOperator{\Hom}{Hom}
\DeclareMathOperator{\Stab}{Stab}

\newcommand{\rsp}{\operatorname{rowsp}}
\newcommand{\rk}{\operatorname{rk}}
\newcommand{\mx}{\mathrm{x}}

\def\cC{ {\cal C} }

\newcommand{\then}{\Longrightarrow}

\usepackage{tikz-cd} 
\usetikzlibrary{babel} 

\makeatletter
\def\fixtikzforbreqn#1#2{%
  \protected\edef#1{\noexpand\ifmmode\mathchar\the\mathcode`#2 \noexpand\else#2\noexpand\fi}%
}
\fixtikzforbreqn\tikz@nonactivesemicolon;
\fixtikzforbreqn\tikz@nonactivecolon:
\fixtikzforbreqn\tikz@nonactivebar|
\fixtikzforbreqn\tikz@nonactiveexlmark!
\makeatother

\usepackage{breqn}

\begin{document}

\title{\bf New codes from algebras of direct products of groups}

\author{\sc Markus Grassl\thanks{International Centre for Theory of
    Quantum Technologies (ICTQT), University of Gda{\'n}sk, ul. Prof. Marii Janion 4, 80-309 Gdańsk, Poland.  ORCID: 0000-0002-3720-5195}\; · {\sc Miguel Sales-Cabrera}$^{\dagger}$\; · {\sc Xaro Soler-Escrivà}\thanks{Departament de Matem\`{a}tiques, Universitat d'Alacant, Ap.\ Correus 99, E-03080 Alacant, Spain. ORCIDs: 0009-0001-8830-3473, 0000-0001-7595-7032}\\  \sc Víctor Sotomayor\thanks{Departamento de Álgebra, Universidad de Granada, Avenida de Fuentenueva s/n, 18071 Granada, Spain. ORCID: 0000-0001-8649-5742 \newline
\Letter: \texttt{markus.grassl@ug.edu.pl, miguel.sales@ua.es, xaro.soler@ua.es, vsotomayor@ugr.es} \newline \rule{6cm}{0.1mm}\newline
This study forms part of the Quantum Communication programme and was
supported by MCIN with funding from European Union NextGenerationEU
(PRTR-C17.I1) and by Generalitat Valenciana  COMCUANTICA/008. This
work is also partially supported by the Ministerio de Ciencia e
Innovaci\'on project PID2022-142159OB-I00. Xaro Soler-Escrivà is
supported by the Generalitat Valenciana project CIAICO/2022/167. Miguel Sales-Cabrera and Xaro Soler-Escrivà are supported by the I+D+i projects
VIGROB23-287 and UADIF23-132 of Universitat d'Alacant. Víctor Sotomayor
was supported by Ayuda a Primeros Proyectos de Investigación
(PAID-06-23) from Vicerrectorado de Investigación de la Universitat
Politècnica de València (UPV). Markus Grassl is supported by grant
no.~FENG.02.01-IP.05-0006/23, financed by the European Funds for a
Smart Economy 2021-2027 (FENG), Priority FENG.02 Innovation-friendly
environment, Measure FENG.02.01. Part of the calculations were carried
out on computers funded under contract no. 2018/MAB/5/AS-1 and at the
Max Planck Institute for the Science of Light, Erlangen.
}}

{\small \date{\today}} 

\maketitle

\begin{abstract}
This paper deals with the algebraic-combinatorial description of group codes by using the Wedderburn-Artin decomposition of semisimple group algebras. Concretely, we give a detailed description of the Wedderburn-Artin decomposition of the semisimple group algebra associated with a direct product of finite groups that are cyclic, dihedral, or generalised quaternion. This, in turn, provides a group-algebraic framework for computing all linear codes generated by such groups. As a consequence, we find new linear codes that improve the best-known minimum distance for their length and dimension. We conclude the paper by discussing how group codes arised from some explicit Wedderburn-Artin decomposition can be utilised to construct quantum codes in a systematic way.


\bigskip


\end{abstract}

\section{Introduction}

Let $\Fq$ be a finite field of $q$ elements, where $q$ is a prime
power, and let $G$ be a finite group of order $n$. In the setting of
linear coding theory, a group code, or simply a $G$-code, over $\F_q$
of length $n$ is a linear code such that its image under the
isomorphism that maps the canonical basis of $\Fq^n$ to the elements
of $G$ is an ideal in the group algebra $\Fq[G]$.  A considerable
amount of codes, such as generalised Reed-Solomon codes, Reed-Muller
codes, and many other optimal codes, have been shown to be group codes
(see for instance \cite{bernal2009intrinsical, borello2021dihedral,
  landrock1992, mcloughlin2008} and the references therein). The
interest in the study of group codes is clearly linked to their
powerful algebraic structure, which allows valuable information about
the parameters of the code to be obtained, as well as to provide
efficient coding and decoding algorithms (\emph{cf.}
\cite{martinez2023}, \cite{landrock1992}).


Most of the research carried out to date has focused on the case where the group $G$ is abelian, resulting in group codes that are necessarily two-sided ideals in the group algebra. However, in recent years the study of the non-abelian case has been gaining an increasing interest (\cite{borello2021dihedral, Gao2021, gao2020lcd}, to name a few). There are several reasons for this. First, a non-abelian group algebra has a richer algebraic structure. In particular we can deal with not necessarily two-sided ideals, so we can construct linear codes that cannot be obtained using abelian groups. 
Second, non-commutativity could possibly improve the security of
code-based encryption protocols, which are one of the few mathematical
techniques that enables the construction of public-key cryptosystems
that are likely to be secure against an adversary equipped with a quantum computer (\emph{cf.} \cite{Sendrier}). This paper is intended to be a new contribution to the non-abelian case research and, for this reason, when we talk about ideals of $\Fq[G]$ we always assume that they are left ideals. Thus, a $G$-code is a linear code of $\Fq^n$ isomorphic to a left ideal  of $\Fq[G]$.


Let us assume hereafter that the characteristic of the field and the
group order are relatively prime. Hence the group algebra
$\mathbb{F}_q[G]$ is semisimple by a celebrated theorem due to
Maschke. Moreover, as a consequence of the Wedderburn-Artin theorem
and the Wedderburn little theorem, $\Fq[G]$ is isomorphic to a direct
sum of some matrix rings over finite extensions of $\Fq$ (\emph{cf.}
\cite{doerk1992finite}). Thus the ideals of $\mathbb{F}_q[G]$ are
principal and can be seen as sum of matrix ideals over finite
fields. It turns out that all possible $G$-codes can be determined
whenever the Wedderburn-Artin decomposition of the group algebra
$\mathbb{F}_q[G]$ is known. In particular, the generators of these
codes can be realised as matrices over finite fields. In general,
obtaining such a decomposition explicitly is a challenging problem. In
the last decade, this has been done for certain non-abelian groups,
such as (generalised) dihedral, (generalised) quaternion, metacyclic,
symmetric and alternating groups (\emph{cf.} \cite{Brochero2015,
  gao2020lcd, Gao2021, Brochero2022, Ricardo2023}, respectively).  

The aim of this paper is to take a step further in the aforementioned research line, contributing to our understanding of both the parameters and the structure of linear codes that are group codes. More concretely,  we focus on the Wedderburn-Artin decomposition of the semisimple group algebra of any direct product of finite groups $G_1\times \cdots \times G_r$ in terms of the structure of the corresponding group algebras of the direct factors $G_1, \ldots, G_r$. As an immediate consequence, we get the full description of semisimple group algebras of direct products of dihedral, cyclic, or generalised quaternion groups; moreover, we provide some new linear codes, which arise from semisimple algebras of direct products of the mentioned groups, that improve the best-known minimum distance for their lengths and dimensions. It is worthwhile to remark that our study is general, in the sense that it does not depend on either the number of factors in the direct product $G_1\times \cdots \times G_r$ or any arithmetical conditions on those factors $G_1, \ldots, G_r$: provided that the decomposition of the factor algebras $\Fq[G_i]$ is explicitely known, we can use algebraic-combinatorial techniques to obtain a detailed expression for the algebra of their product $\Fq[G_1\times \cdots \times G_r]$. In particular, the results obtained in \cite{Vedenev2019_2}, where the authors considered group codes in the group algebra $\Fq[D_n\times D_m]$ of a direct product of two dihedral groups of orders $2n$ and $2m$, respectively, such that $m$ divides $q-1$, appear now as particular cases.


The paper is organised as follows. In Section \ref{sec_prel} we
collect some preliminary results concerning tensor products and the
Wedderburn-Artin decomposition of semisimple group algebras. In
Section \ref{sec_groups}, after recalling some basic facts concerning
group codes, we provide a combinatorial formula for the number of
group codes that can be constructed from the Wedderburn-Artin
decomposition of a given group algebra, as well as a formula for their
dimensions. After that, the general structure of a semisimple algebra
associated to a direct product of finite groups is determined in
Section \ref{sec_direct}. As a consequence, the full description of
the group algebra associated to the direct product of either a
dihedral group and a cyclic group, two dihedral groups, or a dihedral
group and a generalised quaternion group is obtained. In Section
\ref{sec_examples} we show how the theoretical results stated in this
paper can be applied; concretely, we build some group codes arising
from algebras of direct products of the above mentioned groups which,
as linear codes, beat the best-known minimum distance for their
lengths and dimension. Finally, in the concluding section of the
paper, we briefly discuss how group codes arised from a known
Wedderburn-Artin decomposition can be utilised to construct quantum
codes in a systematic way.

\section{Preliminaries}
\label{sec_prel}

All groups considered in this paper are supposed to be finite. We denote by $\Fq$ the finite field of $q$ elements, where $q$ is a power of a prime $p$, whereas $\mathbb{K}$ will denote an arbitrary field. Given an arbitrary ring $R$, we write $M_{n}(R)$ for the ring of $(n\times n)$-matrices over $R$. If $X$ is a non-empty subset of $R$, then $\langle X\rangle_R$ is the ideal of $R$ generated by $X$ (we will simply write $\langle X\rangle$ when the ambient ring $R$ is clear enough).  Moreover, all algebras and rings considered in this paper are associative and unitary. The remaining unexplained notation and terminology are standard in the context of coding theory and group theory.

\subsection{On tensor products}
Below, several properties about tensor products that will be of importance for the development of this paper are listed. They can be found in many books covering tensor products of algebras, such as \cite{Bourbaki1973, Adkins1992, doerk1992finite}.

Recall that given a commutative ring $R$, an $R$-algebra $A$ is a ring which is also an $R$-module. In particular, in what follows, all $R$-algebras are always unitary rings. Moreover, $A$ will be a commutative $R$-algebra if $A$ is a commutative ring.  An homomorphism of $R$-algebras is an $R$-linear ring homomorphism. Given two $R$-algebras $A$ and $B$, we will denote by $A\otimes_R B$ its tensor product over $R$. Defining the product on elements of the form $a\otimes_R b$ by $(a_1\otimes_R b_1)(a_2\otimes_R b_2)=a_1a_2\otimes_R b_1b_2$, it turns out that $A\otimes_R B$ is an $R$-algebra too.  Moreover,  the tensor product of $R$-algebras is associative and it is commutative whenever $A$,  $B$ and $R$ are commutative. Next we list some other properties on tensor product of $R$-algebras that will be needed in the sequel.  As usual, we will denote by $\oplus$  the direct sum of $R$-modules.

\begin{lemma}\label{tpproperties}
Let $R$ be a commutative ring.
\begin{propenum}

    \item Assume that  $M_i, N_i$ are $R$-algebras such that $M_i \cong N_i$ as $R$-algebras, for $i=1,2$. Then $M_1 \otimes_R M_2 \cong N_1 \otimes_R N_2$ as $R$-algebras.
    
    \item\label{tpproperties:directsum} Given two $R$-algebras $M$ and
      $N$ such that  $M \cong \bigoplus_{i \in I} M_i$ and $N \cong
      \bigoplus_{j \in J} N_j$ as $R$-algebras, for some families of
      $R$-algebras $\{M_i\colon i\in I\}$ and $\{N_j\colon j\in J\}$,
      then the following isomorphism of $R$-algebras holds: $$M
      \otimes_R N \cong \bigoplus_{(i, j) \in I \times J} \left( M_i
      \otimes_R N_j\right).$$ 
    
    \item\label{tpproperties:function} Given two $R$-algebra
      homomorphisms $f\colon M_1 \longrightarrow N_1$ and $g\colon M_2
      \longrightarrow N_2$, the map
    \[
    \begin{array}{cccc}
    f \otimes_R g\colon & M_1 \otimes_R M_2 & \longrightarrow  & N_1 \otimes_R N_2\\ 
    & a \otimes_R b & \longmapsto &f(a) \otimes_R g(b)
    \end{array}  
    \]
    defines a homomorphism of $R$-algebras. As a consequence, we obtain the canonical homomorphism of $R$-modules   
    \begin{equation}\label{Hom_lambda}
        \lambda\colon \Hom_R(M_1, N_1) \otimes_R \Hom_R(M_2, N_2) \longrightarrow \Hom_R(M_1 \otimes_R M_2, N_1 \otimes_R N_2).
    \end{equation}
\end{propenum}
\end{lemma}

In the particular case when $R$ is a field, the homomorphism $\lambda$ given in \eqref{Hom_lambda} can be used to produce the following well-known result concerning the tensor product of matrices over a field. Although its proof is elementary, since we were unable to find it stated in this way, we include it here for the sake of completeness.

Given a $\mathbb{K}$-algebra $A$, we will denote by $A^{\oplus n}$ the direct sum of $n$ copies of $A$. The set of all $A$-linear endomorphisms of $A^{\oplus n}$ will be denoted as  $\End_{A}\left(A^{\oplus n}\right)$. Notice that $\End_{A}\left(A^{\oplus n}\right)$ is, in particular, an $A$-algebra.

\begin{lemma}\label{tpmatrices}
    Let $A_1, A_2$ be two commutative algebras over a field $\mathbb{K}$. Then the following isomorphism of $\mathbb{K}$-algebras holds: 
    \[
    M_{n_1}\left(A_1\right) \otimes_{\mathbb{K}} M_{n_2}\left(A_2\right) \cong M_{n_1n_2} \left( A_1 \otimes_{\mathbb{K}} A_2 \right).
    \]
\end{lemma}

\begin{proof}
    Since $\mathbb{K}$ can be embedded into $A_1$ and $A_2$, it turns out that $\End_{A_i}\left(A_i^{\oplus n_i}\right)$ can be realised as a $\mathbb{K}$-subalgebra of $\End_{\mathbb{K}}\left(A_i^{\oplus n_i}\right)$,  for $i=1, 2$. Therefore, according to \cite[Lemma 1, p.~214]{Bourbaki2023}, the homomorphism $\lambda$ defined in \eqref{Hom_lambda} induces the isomorphism of $\mathbb{K}$-vector spaces
    \[
    \varphi\colon \End_{A_1}\left(A_1^{\oplus n_1}\right) \otimes_{\mathbb{K}} \End_{A_2}\left(A_2^{\oplus n_2}\right) \longrightarrow  \End_{A_1 \otimes A_2}\left(A_1^{\oplus n_1} \otimes_{\mathbb{K}} A_2^{\oplus n_2}\right).
    \]
    Let us see that $\varphi$ is a ring homomorphism too. 
    Let $f_1, g_1 \in \End_{A_1}\left(A_1^{\oplus n_1}\right)$ and $f_2, g_2 \in \End_{A_2}\left(A_2^{\oplus n_2}\right)$. For any $a \in A_1^{\oplus n_1}$ and $b \in A_2^{\oplus n_2}$, applying \cref{tpproperties:function}, one gets%
    \[\arraycolsep=2pt \def\arraystretch{1.5}
    \begin{array}{rcl}
       \left((f_1 \otimes_{\mathbb{K}} f_2)\circ (g_1 \otimes_{\mathbb{K}} g_2)\right)(a \otimes_{\mathbb{K}} b) & = & (f_1 \otimes_{\mathbb{K}} f_2)(g_1(a) \otimes_{\mathbb{K}} g_2(b))  \\
       & = & f_1(g_1(a)) \otimes_{\mathbb{K}} f_2(g_2(b))  \\
        & = & (f_1 \circ g_1)(a) \otimes_{\mathbb{K}} (f_2 \circ g_2)(b) \\ 
        & = & ((f_1 \circ g_1) \otimes_{\mathbb{K}} (f_2 \circ g_2))(a \otimes_{\mathbb{K}} b),
    \end{array}
    \]
    so $\varphi$ is in fact a $\mathbb{K}$-algebra isomorphism. Now, taking into account that $\End_{A_i}\left(A_i^{\oplus n_i}\right) \cong M_{n_i}\left(A_i\right)$ for $i=1, 2$ (see \cite[Chapter 4, Corollary 3.9, p.~219]{Adkins1992}), 
    and that $A_1^{\oplus n_1} \otimes_{\mathbb{K}} A_2^{\oplus n_2} \cong \left(A_1 \otimes_{\mathbb{K}} A_2\right)^{\oplus n_1n_2}$ (\cref{tpproperties:directsum}), the result follows.
\end{proof}

The next explicit description of the tensor product of two finite fields as a direct sum of fields will be also needed for our purposes. We also give here an elementary proof for the sake of completeness, although it can be deduced, for instance, from \cite{Cohn}.

Recall that, given 
an extension of finite fields $\Fqr{t}/\Fq$, its Galois group is a cyclic group of order $t$, namely, $\gal(\Fqr{t}/\Fq) = \langle \sigma \rangle$, where $\sigma(x) = x^{q}$, for all $x\in \Fqr{t}$.

\begin{proposition}\label{tpfields}
Let $\Fqr{n}$ and $\Fqr{m}$ be two finite fields, and denote $d =
\gcd(n, m)$ and $\ell = \lcm(n, m)$. Then, as an
$\F_{q^\ell}$-module, the tensor product $\Fqr{n} \otimes_{\Fq} \Fqr{m}$ is isomorphic to the direct sum of $d$ copies of the field $\Fqr{\ell}$, that is, 
\begin{equation}\label{eqtpfields}
    \Fqr{n} \otimes_{\Fq} \Fqr{m}\cong \left(\Fqr{\ell}\right)^{\oplus d}.
\end{equation}

\end{proposition}

\begin{proof}
    Let $\alpha, \beta$ such that $\Fqr{n} = \Fq(\alpha)$ and $\Fqr{m} = \Fq(\beta)$ and take into account the following field extension diagram:

\begin{center}
\begin{tikzcd}
& {\Fqr{\ell} = \Fq(\alpha, \beta)} \arrow[rd, "d_1", no head] \arrow[dd, no head] \arrow[ld, "d_2"', no head] & \\
\Fqr{m} = \Fq(\beta) \arrow[rd, "d_1"', no head] \arrow[rdd, "{[\Fqr{m}:\Fq] = m}"', no head, bend right] & & \Fqr{n} = \Fq(\alpha) \arrow[ld, "d_2", no head] \arrow[ldd, "{[\Fqr{n}:\Fq] = n}", no head, bend left] \\
& \Fqr{d} \arrow[d, "d", no head] & \\ & \Fq &
\end{tikzcd}
\end{center}
Notice that $\ell d=nm$ and there exist integers $d_1,d_2$ such that $\ell=d_1 n=d_2 m$, $n=d_2d$, $m=d_1d$. 

If we consider the Galois group $\gal(\Fqr{\ell} / \Fq) = \langle \sigma \rangle$ of the extension  $\Fqr{\ell} / \Fq$, then it holds $\gal(\Fqr{\ell} / \Fqr{n}) = \langle \sigma^n\rangle$.
Let $p(\mx) \in \Fq[\mx]$ be the minimal polynomial of $\beta$ over $\Fq$. We are going to consider the action of $\langle \sigma^n\rangle$ on the set of roots of $p(\mx)$, which is $R_p = \left\{ \beta, \beta^{q}, \beta^{q^2}, \dots, \beta^{q^{m-1}}\right\}\subseteq \Fqr{m}$. Since the order of $\sigma^n$ is $d_1$, its action on $R_p$ will provide $d$ orbits of $d_1$ elements each. To see this, first observe that, $\langle \sigma^m \rangle\leq \Stab_{\langle\sigma\rangle}(\omega)$, for any root  $\omega$ of $p(\mx)$. Thus, taking $k_1, k_2 \in \mathbb{Z}$ such that $d=nk_1+mk_2$, one has
\[
\omega^{q^d} = \sigma^d(\omega) = \sigma^{nk_1+mk_2}(\omega) = \sigma^{nk_1}\left(\sigma^{mk_2}(\omega)\right) = \sigma^{nk_1}(\omega)\in \mathcal{O}_{\langle \sigma^n \rangle}(\omega),
\]
where $\mathcal{O}_{\langle \sigma^n \rangle}(\omega)$ denotes the
orbit of $\omega$ with respect to the group generated by $\sigma^n$.
Therefore, we obtain $d_1$ different elements in this orbit and, as a result, we conclude that  
\[
 \left\{ \omega^{q^{rd}}\colon 0 \leqslant r \leqslant d_1-1 \right\} = \mathcal{O}_{\langle \sigma^n \rangle}(\omega)= \left\{ \sigma^{rn}(\omega)\colon 0 \leqslant r \leqslant d_1-1 \right\}.
\]

Let $\omega_k = \beta^{q^{k-1}} \in \Fqr{\ell}$
for $k=1, \dots, d$. Then $R_p=\bigcup_{k=1}^d \mathcal{O}_{\langle \sigma^n \rangle}(\omega_k)$ and we can write $p(\mx) = \prod_{k=1}^d p_k(\mx)$, where every $p_k(\mx) \in \Fqr{n}[\mx]$ is irreducible over $\Fqr{n}$, has degree $d_1$ and roots $R_{p_k}=\mathcal{O}_{\langle \sigma^n \rangle}(\omega_k)$.

 


Now, consider the following map
\[\arraycolsep=2pt
\begin{array}{rccl}
  \nu\colon & \Fqr{n} \otimes_{\Fq} \Fqr{m} & \longrightarrow & \left(\Fqr{\ell}\right)^{\oplus d}\\[1mm]
  & \alpha^i \otimes_{\Fq} \beta^j & \longmapsto & \alpha^i \big( \omega_1^j, \dots, \omega_d^j \big)=\big( \alpha^i \omega_1^j, \dots, \alpha^i \omega_d^j \big).
\end{array}
\]
Since $\left\{ \alpha^i \otimes \beta^j\colon 0 \leqslant i \leqslant n-1, 0 \leqslant j \leqslant m-1 \right\}$ is an $\Fq$-basis of $\Fqr{n} \otimes_{\Fq} \Fqr{m}$, extending by $\Fq$-linearity, one has that $\nu$ defines an  $\Fq$-vector space homomorphism. Moreover, it can be easily checked that $\nu$ respects multiplication. Thus, $\nu$ is an  $\Fq$-algebra homomorphism. We will show that it is indeed an isomorphism, and the proof will be complete.

Let $u = \sum_{i, j} \lambda_{ij} (\alpha^i \otimes \beta^j) \in \ker(\nu)$, where $\lambda_{ij} \in \Fq$. Then
\begin{equation}\label{eqdemotensor}
    \nu(u) = \sum_{i, j} \lambda_{ij} \alpha^i \big( \omega_1^j, \dots, \omega_d^j \big) = (0,\dots, 0).
\end{equation}

Consider the polynomial $P_\alpha(\mx) = \sum_{i,j} \lambda_{ij} \alpha^i \mx^j \in \Fqr{n}[\mx]$, which satisfies $\deg(P_\alpha(\mx)) \leqslant m-1$. By \eqref{eqdemotensor}, we obtain that $P_\alpha(\omega_k) = 0$, for all $k \in \{1, \dots, d\}$. Therefore, every $p_k(\mx)$ must divide $P_\alpha(\mx)$, and so must its product $\prod_{k=1}^d p_k(\mx) = p(\mx)$. If $P_\alpha(\mx)$ is not the zero polynomial, then the following contradiction would be reached: $m = \deg(p(\mx)) \leqslant \deg(P_\alpha(\mx)) \leqslant m-1$. Thus, $P_\alpha(\mx)$ must be the zero polynomial and then $\sum_{i=1}^{n-1} \lambda_{ij} \alpha^i = 0 $, for all $j \in \{0, \dots, m-1\}$. Since $\{ \alpha^i \colon 0 \leqslant i \leqslant n-1\}$ is an $\Fq$-basis of $\Fqr{n}$, it can only happen that $\lambda_{ij} = 0$, for all $i \in \{0, \dots, n-1\}$ and $j \in \{0, \dots, m-1\}$. This means that $\ker(\nu) = \{0\}$, that is, $\nu$ is injective.

In addition, $\dim_{\Fq}(\Fqr{n} \otimes_{\Fq} \Fqr{m}) =
\dim_{\Fq}(\Fqr{n})\dim_{\Fq}(\Fqr{m}) = nm=\ell d=\dim_{\Fq}\left(\left(\Fqr{\ell}\right)^{\oplus d}\right)$. Consequently, $\nu$ is an isomorphism of $\Fq$-algebras.
\end{proof}

\subsection{On the decomposition of a group algebra}\label{sec:decomposition}
\label{group-algebra-decomposition}

Given a group $G$, the set of all formal $\mathbb{F}_q$-linear combinations of elements of $G$, \emph{i.e.}
\[
\Fq[G]=\left\{\displaystyle\sum_{g\in G} \alpha_g g\colon \alpha_g\in \Fq \right\},
\]
is an $\Fq$-vector space with the elements of $G$ forming a basis. Moreover, by considering the multiplication
\[
\left(\sum_{g\in G}\alpha_g g \right)\cdot \left(\sum_{g\in G}\beta_g g\right)=\sum_{g\in G}\left(\sum_{h\in G}\alpha_h\beta_{h^{-1}g}\right) g, 
\]
we obtain an $\Fq$-algebra which is called the \emph{group algebra} of $G$ over $\Fq$.

In this paper we will always deal with group algebras $\Fq[G]$ that are semisimple, that is, they can be realised as a direct sum of simple $\Fq[G]$-modules. Recall that Maschke's theorem states that $\Fq[G]$ is a semisimple $\Fq$-algebra if and only if the characteristic of the field does not divide the order of $G$ (\emph{cf.} \cite{doerk1992finite,Bourbaki2023}).
 This is the reason why, from now on, the order of the considered groups $G$ is not divisible by the characteristic of the field $\Fq$.  Moreover, when we talk about ideals of $\Fq[G]$ we always assume that they are left ideals. 
 
Semisimple $\Fq$-algebras have many important properties. For instance, assuming that $\Fq[G]$ is semisimple, every ideal of $\Fq[G]$ is always a direct summand. As a consequence, every ideal of the algebra must be principal. Moreover, the algebra can be realised as a direct sum of matrix algebras. This is the well-known Wedderburn-Artin theorem,  which is stated below for finite group algebras (see \cite[Theorem 4.4, p.~112]{doerk1992finite}).

\begin{theorem}[Wedderburn-Artin decomposition for finite group algebras]\label{Wedderburn}
    Let $G$ be a finite group such that $\Fq[G]$ is a semisimple group algebra. Then $\Fq[G]$ is isomorphic, as $\Fq$-algebra, to the direct sum of some matrix rings over suitable extensions of $\Fq$. Specifically, one has: 
    \[
        \Fq[G] \cong \bigoplus_{i=1}^s M_{n_i} \left( \Fqr{r_i} \right)
    \]
    satisfying $|G| = \sum_{i=1}^{s} n_i^2 r_i$.
\end{theorem}


For some semisimple group algebras over finite fields, the Wedderburn-Artin decomposition is known. Below we list some of them:
\begin{itemize}
    \item Let $C_n$ be the cyclic group of order $n$. Assume that the polynomial $\mx^n-1\in\Fq[\mx]$ can be factorised as $\mx^n-1 = \prod_{j=1}^r f_j$, where $f_j$ is irreducible over $\Fq[\mx]$. Then the Wedderburn-Artin decomposition of $\Fq[C_n]$ can be obtained by applying the Chinese Remainder Theorem: 
    \begin{equation}\label{eq:descomp_ciclic}
    \Fq[C_n] \cong \dfrac{\Fq[\mx]}{\langle\mx^n-1\rangle} \cong  \bigoplus_{i=1}^{r} \dfrac{\Fq[\mx]}{\langle f_i\rangle} \cong \bigoplus_{i=1}^{r} \Fqr{\deg{f_i}}.
  \end{equation}
\end{itemize}

The generalisation of the previous result for abelian groups is computed in \cite{Perlis1950}, which is the so-called Perlis-Walker Theorem. 

\begin{itemize}
    \item If $G$ is an abelian group of order $n$, then the Wedderburn-Artin decomposition of $\Fq[G]$ is as follows:
    \begin{equation}\label{eq:descomp_abelia}
    \Fq[G] \cong \bigoplus_{d | n} (\mathbb{F}_{q^{t_d}})^{\oplus a_d},
    \end{equation}
    where $t_d=[\Fq(\alpha_d) : \Fq]$, with $\alpha_d$ a primitive $d$-th root of unity, $a_d = \frac{n_d}{t_d}$, and $n_d$ the number of elements of order $d$ in $G$. We use the convention that a direct summand indexed by some $d$ with $n_d=0$ is omitted.

\end{itemize}

In \cite[Theorem 3.1]{Brochero2015}, the Wedderburn-Artin
decomposition of dihedral group algebras is given.  For every
polynomial $g\in \Fq[\mx]$ with non-zero constant coefficient, we
denote by $g^*$ the {\em reciprocal} polynomial of $g$, \emph{i.e.}
$g^*(\mx) =g(0)^{-1}\mx^{\deg(g)}g(\mx^{-1})$.  The polynomial $g$ is
said to be {\em self-reciprocal} if $g=g^*$.  In this case,  $g$
always has even degree, unless $g(\mx)=\mx\pm 1$ (see \cite[Remark 3.2]{Brochero2015}).
Consequently $\mx^n-1$ can be factorised over $\Fq[\mx]$ into irreducible monic polynomials as 
\[
\mx^n-1 = f_1f_2 \cdots f_r f_{r+1}f_{r+1}^* \cdots f_{r+s}f_{r+s}^*,
\]
where $f_1 = \mx-1$, $f_2=\mx+1$ if $n$ is even, and $f_j=f_{j}^*$ for $1\leqslant j\leqslant r$. In this way, $r$ is the number of self-reciprocal factors in the factorisation and $2s$ the number of factors that are not self-reciprocal.

\begin{itemize}
    \item Let $D_n=\langle x, y\,|\, x^{n} =y^2= 1, y^{-1}xy = x^{-1} \rangle$ be the dihedral group of order $2n$. Set $\zeta(n) = 2$ if $n$ is even and $\zeta(n) = 1$ otherwise.  Then
    \begin{equation}\label{eq:descomp_diedric}
    \Fq[D_n] \cong \bigoplus_{i=1}^{r+s}A_i , \qquad\text{where}\quad
    A_i= 
    \begin{cases}
        \Fq \oplus \Fq, & \text{if $1 \leqslant i \leqslant \zeta(n)$;} \\[1mm]
        M_2 \left(\Fqr{\deg(f_i)/2} \right), & \text{if $\zeta(n) +1\leqslant i \leqslant r$;} \\[2mm]
         M_2 \left(\Fqr{\deg(f_i)} \right), & \text{if $r+1 \leqslant i \leqslant r + s$.}
    \end{cases}
    \end{equation}

\end{itemize}

Pursuing this line of work, in \cite[Theorems 3.1 and 3.6]{Gao2021}, one can find the  Wedderburn-Artin decomposition for generalised quaternion group algebras. In this case, we need to consider both the factorisation into irreducible monic polynomials over $\Fq[\mx]$ of $\mx^n-1$ given in the dihedral case, and also the factorisation into irreducible monic polynomials of the polynomial $\mx^n+1 = g_1g_2 \cdots g_t g_{t+1}g_{t+1}^* \cdots g_{t+k}g_{t+k}^*$, where $g_1 = \mx+1$ if $n$ is odd. 

\begin{itemize}
       
    \item Let $Q_{n}$ be the generalised quaternion group of order $4n$, with $n\geq 1$, which admits the presentation $Q_{n} = \langle x, y\,|\, x^{2n} = 1, y^2 = x^n, y^{-1}xy = x^{-1} \rangle$.  Set $\mu(n) = 0$ if $n$ is even and $\mu(n) = 1$ otherwise.  Then 
    \begin{equation}\label{eq:descomp_quaterni}
    \Fq[Q_{n}] \cong \bigoplus_{i=1}^{r+s}A_i\oplus \bigoplus_{i=1}^{t+k}B_i,
        \end{equation}
    where every $A_i$ is given as in the dihedral case (see (\ref{eq:descomp_diedric})) and 
    \[
    B_i= 
    \begin{cases}
        \Fq \oplus \Fq, & \text{if $1\leq  i \leqslant \mu(n)$ and } q\equiv 1 \pmod{4};\\[1mm]
        \Fq(\sqrt{-1}), & \text{if $1\leq  i \leqslant \mu(n)$ and } q\equiv 3 \pmod{4}; \\[1mm]
        M_2\left(\Fqr{\deg(g_i)/2}\right), & \text{if $\mu(n) + 1\leqslant i \leqslant t$;} \\[2mm]
         M_2\left(\Fqr{\deg(g_i)}\right), & \text{if $t+1 \leqslant i \leqslant t + k$.}
    \end{cases}
    \]
\end{itemize}

It is worth noting that there exists a criterion to decide when the group algebras $\Fq[D_{2n}]$ and $\Fq[Q_{n}]$ are isomorphic or not. This is done in \cite{dutra2009semisimple}, where it is shown that $\Fq[D_{2n}]$ and $\Fq[Q_{n}]$ are isomorphic $\Fq$-algebras if and only if $n$ is even or $q\equiv 1\pmod 4$.

There are other groups for which the Wedderburn-Artin decomposition over finite fields is also known. For instance, in \cite{gao2020lcd} and \cite{Brochero2022}, the authors compute the Wedderburn-Artin decomposition for generalised dihedral group algebras and some metacyclic group algebras, respectively. Besides, by adapting known results about $\mathbb{Q}[S_n]$ and $\mathbb{Q}[A_n]$, the Wedderburn-Artin decompositions of $\Fq[S_n]$ and $\Fq[A_n]$ are obtained in \cite{Ricardo2023}.


\section{Group codes in semisimple group algebras}
\label{sec_groups}

In coding theory, a linear code $\mathcal{C}$ of $\Fq^n$ is said to be \emph{cyclic} if it satisfies that a word $(c_1, c_2, \dots, c_n)$ lies in $\cC$ if and only if $(c_2,\dots, c_n, c_1)$ also lies in $\cC$. Cyclic codes have nice properties, and efficient decoding algorithms have been developed for them. The key fact is that a cyclic code can be seen as an ideal, generated by a monic polynomial dividing $\mx^n-1$, in the quotient ring $\Fq[\mx]/\langle \mx^n-1\rangle$, which is a principal ideal ring.
The notion of group codes is a natural extension of cyclic codes, so that a
cyclic code is a group code when the associated group is cyclic.

Set $G=\{g_1,\dots ,  g_n\}$. Following \cite{gonzalez2019group, bernal2009intrinsical}, we say that a linear code $\cC\subseteq \Fq^n$ is a {\em $G$-code} if there exists a bijection $\theta\colon \{1,\dots, n\} \longrightarrow G$ such that the set
\[
\left\{\sum_{i=1}^n a_i \theta(i)\colon (a_1, \dots, a_n) \in  \cC \right\}
\]
is a left ideal of the group algebra $\Fq[G]$. In this way, a linear code $\cC$ over $\Fq$ will be a \emph{group code} if there exists a finite group $G$ such that $\cC$ is a $G$-code. In practice, we usually identify $\cC$ with its corresponding ideal of $\F[G]$. When the group $G$ is abelian (resp. non-abelian) we say that $\cC$ is an abelian (resp. non-abelian) group code.  Not all linear codes can be realised as group codes; in fact, in \cite{bernal2009intrinsical} the reader can find a criterion to decide when a linear code is a group code. Moreover, notice that a given linear code $\cC$ can be seen as group code over two different groups: for instance $\cC = \{000000,111111\}\subseteq \mathbb{F}_2^6$ is a $G$-code for any group $G$ of order $6$.




Assume that the $\Fq$-algebra of a group $G$ is semisimple. Then by Theorem \ref{Wedderburn} we can write  $\Fq[G] \cong \bigoplus_{i=1}^s M_{n_i} \left( \Fqr{r_i} \right)$ for some positive integers $n_i, r_i$, and any ideal of $\Fq[G]$ can be seen as a sum $I_1\oplus\cdots\oplus I_s$,  where $I_i$ is an ideal in $M_{n_i} \left( \Fqr{r_i} \right)$, for $i \in \{1, \dots, s\}$.  Moreover,  this decomposition is unique.  This is the reason why,  in order to study the parameters and properties of group codes of $\Fq[G]$,  we start by exploring the ideals of an arbitrary ring of matrices $M_n(\Fqr{t})$, over a finite field $\Fqr{t}$.   

Since $M_n(\Fqr{t})$ is a principal ideal ring,  every ideal $I$ of $M_n(\Fqr{t})$ is generated by a matrix $M\in M_n(\Fqr{t})$, that is,
\begin{equation}\label{eq:idealgenerat}
I=\langle M\rangle=\{XM\colon X\in M_n(\Fqr{t}) \}.
\end{equation}
From the point of view of ring theory,  the rank of the ideal $I$  is just the rank of any generator matrix $M$ of $I$.  In contrast,  if we want to compute the dimension of the ideal $I$ as an $\Fq$-vector space, then we must first take into account the vector space over 
$\Fqr{t}$ generated by the rows of $M$,  that is, $\rsp(M)=\{xM\colon x\in\Fqr{t}^n \}$,  since then  $\rk(M)$ is exactly its dimension, that is,  
\[
\rk(M)=\dim_{\Fqr{t}}(\rsp(M)).
\] 
Now,  note that $M_n(\Fqr{t}) \cong \left(\Fqr{t}^n\right)^{\oplus n}$ as $\Fqr{t}$-vector spaces through the next isomorphism:
 \[
    \begin{array}{cccc}
    \varphi\colon &  M_n(\Fqr{t}) & \longrightarrow &   \left(\Fqr{t}^n\right)^{\oplus n}\\[1mm] 
     & X = \begin{pmatrix} x_1 \\ x_2 \\ \vdots \\ x_n  \end{pmatrix} & \longmapsto & x_1 \oplus x_2 \oplus \cdots \oplus x_n
    \end{array}  
    \]
In this way,  $\varphi(I)=\varphi(\langle M\rangle )=\rsp(M)\oplus\cdots\oplus \rsp(M)$ and therefore $\dim_{\Fqr{t}}(I)=n\rk(M)$. 
As a result,  the dimension of $I=\langle M\rangle$ as an $\Fq$-vector space is given in terms of $\rk(M)$ as follows. 

\begin{lemma}
The dimension of an ideal $I=\langle M\rangle$ of  $M_n(\Fqr{t})$ as $\Fq$-vector space is $\dim_{\Fq}(I)=t n \rk(M)$.
\end{lemma}

Now we are ready to compute the dimension of any $G$-code when $\Fq[G]$ can be realised as a sum of matrix algebras over finite fields. 

\begin{theorem}\label{code_dim}
    Let $G$ be a group such that $\Fq[G]\overset{\scriptscriptstyle \psi}{\cong} \bigoplus_{i=1}^s M_{n_i} \left( \Fqr{r_i} \right)$ for some positive integers $n_i, r_i$, for $1\leqslant i\leqslant s$, and for some isomorphism $\psi$ of algebras.
    Let     
    $\cC \subseteq \Fq[G]$ be a $G$-code.  If $\psi(\cC) = \bigoplus_{i=1}^s \langle M_i\rangle$ with $M_i\in M_{n_i} \left( \Fqr{r_i} \right)$, then
    \[
    \dim_{\Fq}(\mathcal{C}) = \sum_{i=1}^s n_i r_i \rk(M_i).
    \]
\end{theorem}

The decomposition of the group algebra $\Fq[G]$ into rings of matrices over finite fields not only allows us to give the dimension of any $G$-code,  but also to count the total number of $G$-codes we can construct within this group algebra.  To do so,  note that in an arbitrary ring of matrices $M_n(\Fq)$,  each ideal is generated exactly by one matrix $M\in M_n(\Fq)$ in row reduced echelon form (see \eqref{eq:idealgenerat}).  In turn,  each matrix of $M_n(\Fq)$ in row reduced echelon form having rank $k$ generates exactly a $k$-dimensional vector subspace of $\Fq^n$ (its row space).  Thus,  there is a bijection between the set of ideals of rank $k$ in $M_n(\Fq)$ and  the  set of $k$-dimensional vector subspaces of $\Fq^n$, which is the Grassmann variety ${\cal G}_{q}(k,n)$.  As it is well known, the cardinality of ${\cal G}_{q}(k,n)$ is given by the $q$-ary Gaussian coefficient (\cite[Ch. 24]{van2001course}):
  \[
   \begin{bmatrix} n \\ k \end{bmatrix}_q =\dfrac{(q^n-1)(q^{n-1}-1) \cdots (q^{n-k+1}-1)}{(q^k-1)(q^{k-1}-1) \cdots (q-1)} .
\]
Therefore,  this is also the number of ideals of rank $k$ in $M_n(\Fq)$. 
It follows then that the total number of ideals we can construct in $M_n(\Fq)$, denoted as ${\cal I}_q(n)$,  is exactly the cardinality of the projective geometry ${\cal P}(\Fq^n)$ (that is, the set of all subspaces) associated to the $\Fq$-vector space $\Fq^n$.  This number is:
\[
{\cal I}_q(n)=|{\cal P}(\Fq^n)|=\sum_{k=0}^n  \begin{bmatrix} n \\ k \end{bmatrix}_q  .
\]

Now,  we can compute the exact number of $G$-codes we can construct in $\Fq[G]$ when this algebra is semisimple. 

\begin{theorem}\label{numberofcodes}
    If $\Fq[G] \cong \bigoplus_{i=1}^s M_{n_i} \left( \Fqr{r_i} \right)$, then the number of $G$-codes over $\Fq$ is given by the formula
    \[
     \prod_{i=1}^s {\cal I}_{q^{r_i}}(n_i).
    \]
\end{theorem}

In particular, if $G$ is a dihedral group, then using the decomposition of the semisimple group algebra $\mathbb{F}_q[G]$ given in (\ref{eq:descomp_diedric}) we recover the number $\mathcal{N}$ of dihedral codes given in \cite[Theorem 3.1]{CaoCaoMa}.

\begin{example}
As an easy application of the previous result, and using some of the Wedderburn-Artin decompositions described in Subsection \ref{group-algebra-decomposition}, let us compute the number of group codes in the next group algebras.

\begin{center}
\begin{tabular}{|c|c|c|}
\hline
   Group algebra & Wedderburn-Artin decomposition & Number of group codes \\\hline
$   \mathbb{F}_3[C_5]$ &$ \mathbb{F}_3 \oplus\mathbb{F}_{3^4}$  & $2 \cdot 2 = 4$ \\\hline
  $ \mathbb{F}_3[D_4]$ &$ 4\mathbb{F}_3 \oplus M_2(\mathbb{F}_3)$ &$ 2^4 \cdot 6 = 96$\\\hline
\end{tabular}
\end{center}
\end{example}

\section{Group codes from direct products of groups}
\label{sec_direct}

In this section we focus on group codes that arise in a group algebra of type $\Fq[G_1\times \cdots \times G_r] $,  when the group algebra corresponding to each factor $G_i$ is semisimple. For that aim, starting from the Wedderburn-Artin decomposition of each $\Fq[G_i]$,  we first describe the specific decomposition of the group algebra of the direct product as a direct sum of matrix rings over finite fields (Theorem \ref{Wedderburntp}).   Next we use the known explicit decompositions for group algebras of cyclic, dihedral, and generalised quaternion groups (see equations (\ref{eq:descomp_ciclic}), (\ref{eq:descomp_diedric}), and (\ref{eq:descomp_quaterni}),  respectively) to obtain the corresponding concrete expression for any group that is a direct product of a dihedral group and a cyclic group, the direct product of two dihedral groups, or the direct product of a dihedral group and a generalised quaternion group (Theorem \ref{theo:productes}). In particular, the results of \cite{Vedenev2019_2}, which deal with group codes in the group algebra $\Fq[D_n\times D_m]$ of a direct product of two dihedral groups of orders $2n$ and $2m$, respectively, with $m\mid(q-1)$, are recovered as particular cases of Theorem \ref{theo:productes}.



\subsection{The group algebra of a direct product of groups}

The goal of this section is to obtain the Wedderburn-Artin decomposition of the group algebra corresponding to a direct product of groups having a semisimple finite group algebra. Specifically, given two finite groups $G$ and $H$ such that both $\Fq[G]$ and $\Fq[H]$ can be realised as a direct sum of matrix rings over finite fields (Theorem \ref{Wedderburn}), we are going to see how the tensor product of $\Fq$-algebras can be used in order to obtain the corresponding decomposition of the group algebra $\Fq[G\times H]$ as a  direct sum of matrix rings.  

Note that the map $(g,h)\mapsto g\otimes h$, for all $g\in G$ and $h\in H$ is a homomorphism from $G\times H$ to the group of units of $\Fq[G] \otimes_{\Fq} \Fq[H]$, which can be extended by linearity to an $\Fq$-algebra isomorphism (see \cite[Lemma 3.4, p.~25]{Passman1977}):
\begin{equation}\label{isom_direct_tens}
    \Fq[G \times H ] \cong  \Fq[G] \otimes_{\Fq} \Fq[H] .
\end{equation}


    





\begin{theorem}\label{Wedderburntp}
    Let $\Fq[G]$ and $\Fq[H]$ be two group algebras with the following Wedderburn-Artin decompositions:
    \begin{align*}
        \Fq[G] \overset{\scriptscriptstyle \psi_1}{\cong}  & \bigoplus_{i=1}^{s_G} M_{n_i} \left( \Fqr{r_i} \right) \\
        \Fq[H] \overset{\scriptscriptstyle \psi_2}{\cong} & \bigoplus_{j=1}^{s_H} M_{m_j} \big( \Fqr{t_j} \big).
    \end{align*}
Then, $\Fq[G \times H]$ can be decomposed as:
\[
\Fq[G \times H] \cong \bigoplus_{i=1}^{s_G} \bigoplus_{j=1}^{s_H} \left( M_{n_im_j} \left( \Fqr{\lcm(r_i, t_j)} \right) \right)^{\oplus  \gcd(r_i, t_j)} .
\]

\end{theorem}
\begin{proof}
If we apply the isomorphism provided in \eqref{isom_direct_tens}, joint with \cref{tpfields},  \cref{tpproperties} and \cref{tpmatrices}, then we have the following chain of $\Fq$-algebra isomorphisms:


\begin{eqnarray*}
    \Fq[G \times H] & \overset{\scriptscriptstyle \eqref{isom_direct_tens}}{\cong} &  \Fq[G] \otimes_{\Fq} \Fq[H]\\
    & \overset{\scriptscriptstyle \psi_1 \otimes \psi_2}{\cong} &
        \left( \bigoplus_{i=1}^{s_G} M_{n_i} (\Fqr{r_i}) \right) \otimes_{\Fq} \left(\bigoplus_{j=1}^{s_H} M_{m_j} ( \Fqr{t_j}) \right)\\
    &\overset{\scriptscriptstyle \text{Lem.} \ref{tpproperties}}{\cong} & \bigoplus_{i=1}^{s_G} \bigoplus_{j=1}^{s_H} \left( M_{n_i} \left( \Fqr{r_i} \right) \otimes_{\Fq}  M_{m_j} ( \Fqr{t_j})\right)\\
    &  \overset{\scriptscriptstyle \text{Lem.} \ref{tpmatrices}}{\cong} &
    \bigoplus_{i=1}^{s_G} \bigoplus_{j=1}^{s_H}   M_{n_im_j} \left( \Fqr{r_i} \otimes_{\Fq} \Fqr{t_j}\right)\\
    & \overset{\scriptscriptstyle \eqref{eqtpfields}}{\cong} & \bigoplus_{i=1}^{s_G} \bigoplus_{j=1}^{s_H}   M_{n_im_j} \left((\Fqr{\lcm(r_i, t_j)} )^{\oplus \gcd(r_i, t_j) }  \right)\\
    & \cong & \bigoplus_{i=1}^{s_G} \bigoplus_{j=1}^{s_H}  \left(M_{n_im_j} \left( \Fqr{\lcm(r_i, t_j)} \right)\right)^{\oplus \gcd(r_i, t_j)} .
\end{eqnarray*}
\end{proof}

\begin{remark}
The above result tells us that, if the algebras $\Fq[G]$ and $\Fq[H]$ can be decomposed as direct sums of known matrix algebras over finite fields,  then so can $\Fq[G\times H]$. In fact, after reindexing the terms, we get 
$$
\Fq[G\times H]\cong  \bigoplus_{i=1}^{s_{(G\times H)}} M_{\ell_i} (\Fqr{s_i}),
$$ 
for certain positive integers $s_{(G\times H)}, l_i, s_i$.  Thus,  given a direct product of  finite groups $G_1\times \cdots\times G_r$ such that each algebra $\Fq[G_i]$ has a known explicit decomposition as a direct sum of matrix algebras, the recursive application of Theorem \ref{Wedderburntp} yields the corresponding decomposition of the group algebra $\Fq[G_1\times \cdots \times G_r]$. 
\end{remark}






In particular, using the specific group algebra decompositions listed in Section \ref{sec:decomposition}, we are able to decompose the group algebra of any direct product of groups whose factors are cyclic groups,  dihedral groups or generalised quaternion groups.  In short, what we obtain, for this type of groups, is an expression of the group algebra as a direct sum of rings of matrices over finite fields.  As an example of the application of this technique, below we present the corresponding decomposition for the group algebras of the direct product of a dihedral group and a cyclic group,  the direct product of two dihedral groups or the direct product of a dihedral group and a generalised quaternion group. 

Let us first introduce the following numerical lemma.

\begin{lemma}\label{lem:gcd}
Let $m$, $n$ be two positive integers such that $m$ is even. Denote $\alpha=\gcd(m,n)$ and $\beta=\lcm(m,n)$. Then
\begin{enumerate}[$i)$]
\item If $2\alpha \mid m$, then $\gcd(\frac{m}{2},n)=\alpha$ and $\lcm(\frac{m}{2},n)=\frac{\beta}{2}$.
\item If $2\alpha\nmid m$, then $\gcd(\frac{m}{2},n)=\frac{\alpha}{2}$ and $\lcm(\frac{m}{2},n)=\beta$.
\item If $n$ is also an even number, then $\gcd(\frac{m}{2}, \frac{n}{2})=\frac{\alpha}{2}$ and $\lcm(\frac{m}{2},\frac{n}{2})=\frac{\beta}{2}$.
\end{enumerate}
\end{lemma}

\begin{proof}
$i)$ If  $2\alpha \mid m$, then $\alpha\mid \frac{m}{2}$, since $m$ is even. Therefore $\alpha$ divides $\gcd(\frac{m}{2},n)$, which also divides $\alpha$. Thus, we obtain the first equality.  For the second one,  we notice that
\[
\lcm\left(\frac{m}{2},n\right) \gcd\left(\frac{m}{2},n\right)=\frac{mn}{2} \then \lcm\left(\frac{m}{2},n\right)=\frac{mn}{2\alpha}=\frac{\beta}{2}.
\]
$ii)$ If  $2\alpha \nmid m$, then  $\alpha \nmid \frac{m}{2}$, since $m$ is even. Moreover,  notice that $\alpha$ must be also an even number in this case.  Therefore $\frac{\alpha}{2}$ must divide $ \frac{m}{2}$ and we conclude that $\frac{\alpha}{2}$ is a divisor of $\gcd\left(\frac{m}{2},n\right)$.  In turn,  notice that $\gcd\left(\frac{m}{2},n\right)$ is a proper divisor of $\alpha$. Thus, necessarily, $\frac{\alpha}{2}=\gcd\left(\frac{m}{2},n\right)$. Finally, we observe that
\[
\lcm\left(\frac{m}{2},n\right) \gcd\left(\frac{m}{2},n\right)=\frac{mn}{2} \then  \lcm\left(\frac{m}{2},n\right)=\frac{mn}{\alpha}=\beta.
\]
$iii)$ This follows from well-known properties of gcd and lcm.
\end{proof}

\begin{notation}
\label{notation}
Let us denote by $C_a, D_n, Q_m$ the cyclic group of order $a$, the dihedral group of order $2n$, and the generalised quaternion group of order $4m$, respectively. Set $\zeta(n) = 2$ if $n$ is even and $\zeta(n) = 1$ otherwise; and set $\mu(m) = 0$ if $m$ is even and $\mu(m) = 1$ otherwise. In addition, for each of the previous groups, we consider the following factorisation into irreducible monic polynomials (recall that $f^*$ denotes the reciprocal polynomial of $f$):
\vspace*{-3mm}
\begin{itemize}
\setlength{\itemsep}{-0.5mm}
\item $\mx^a-1 = p_1p_2 \cdots p_b$;
\item $\mx^n-1 = f_1f_2 \cdots f_r f_{r+1}f_{r+1}^* \cdots f_{r+s}f_{r+s}^*$, where $f_1 = \mx-1$, and $f_2=\mx+1$ if $n$ is even.
\item $\mx^m-1 = g_1g_2 \cdots g_t g_{t+1}g_{t+1}^* \cdots g_{t+u}g_{t+u}^*$, where $g_1 = \mx-1$, and $g_2=\mx+1$ if $m$ is even.
\item $\mx^m+1 = h_1h_2 \cdots h_l h_{l+1}h_{l+1}^* \cdots h_{l+k}h_{l+k}^*$, where $h_1=\mx+1$ if $m$ is odd. 
\end{itemize}
\end{notation}

\begin{theorem}\label{theo:productes}
Using Notation \ref{notation}, the following Wedderburn-Artin decompositions hold.
    \begin{corolenum}
        \item \label{corolCxD} Denote $\ell_{ij} = \lcm(\deg(f_i),
          \deg(p_j))$ and $a_{ij} = \gcd(\deg(f_i), \deg(p_j))$ for
          $1\leqslant i\leqslant r+s$ and $1\leqslant j\leqslant
          b$. Then
          \[\Fq[D_n \times C_a] \cong
          \displaystyle\bigoplus_{i=1}^{r+s} \bigoplus_{j=1}^{b}
          \left( A_{ij}^{\oplus d_{ij}} \right),
          \] where 
        \[
        A_{ij} = \begin{cases}
        \Fqr{\deg(p_j)} \oplus \Fqr{\deg(p_j)}, & \text{if $1 \leqslant i \leqslant \zeta(n)$;} \\[1mm]
        M_2 \left(\Fqr{\ell_{ij}/2} \right), & \text{if $\zeta(n)+1\leqslant i \leqslant r$ and $2a_{ij}\mid \deg(f_i)$;} \\[2mm]
        M_2 \left(\Fqr{\ell_{ij}} \right), & \text{if $\zeta(n) +1\leqslant i \leqslant r$ and $2a_{ij} \nmid \deg(f_i)$;} \\[2mm]
        M_2 \left(\Fqr{\ell_{ij}} \right), & \text{if $r+1 \leqslant i \leqslant r + s$;}
        \end{cases}            
        \]
       and
        \[
         d_{ij} = \begin{cases}
           \frac{a_{ij}}{2}, & \text{if $\zeta(n) +1\leqslant i \leqslant r$ and $2 a_{ij} \nmid \deg(f_i)$;} \\[1mm]
           a_{ij}, & \text{otherwise.}
        \end{cases}
        \]

        \item \label{corolDxD} Denote $\ell_{ij} = \lcm(\deg(f_i),
          \deg(g_j))$ and $a_{ij} = \gcd(\deg(f_i), \deg(g_j))$  for
          $1\leqslant i\leqslant r+s$ and $1\leqslant j\leqslant t+u$.
          Then 
        \[
        \Fq[D_n \times D_m] \cong \bigoplus_{i=1}^{r+s} \bigoplus_{j=1}^{t+u} \left( A_{ij}^{\oplus d_{ij}} \right),
        \]
        where
        \[
        A_{ij} = \begin{cases}
        \left(\Fq\right)^{\oplus 4},  & \text{if $1 \leqslant i \leqslant \zeta(n)$ and $1 \leqslant j \leqslant \zeta(m)$;} \\[2mm]
        
        \multirow{2}{*}{$M_2 \left(\Fqr{\ell_{ij}/2} \right)$,} & \text{if $\zeta(n) +1\leqslant i \leqslant r$ and $1 \leqslant j \leqslant \zeta(m)$;} \\ & \text{if $1 \leqslant i \leqslant \zeta(n)$ and $\zeta(m) +1\leqslant j \leqslant t$;} \\[2mm]
        
        \multirow{2}{*}{$M_2 \left(\Fqr{\ell_{ij}} \right)$,} & \text{if $r + 1\leqslant i \leqslant r + s$ and $1 \leqslant j \leqslant \zeta(m)$;} \\ & \text{if $1 \leqslant i \leqslant \zeta(n)$ and $t +1\leqslant j \leqslant t + u$;} \\[2mm]

        \multirow{3}{*}{$M_4 \left(\Fqr{\ell_{ij}/2} \right)$,} & \text{if $\zeta(n) +1\leqslant i \leqslant r$ and $\zeta(m) +1\leqslant j \leqslant t$;} \\ & \text{if $\zeta(n) +1\leqslant i \leqslant r, \  t+1\leqslant j\leqslant t+u$ and $2a_{ij} \mid  \deg(f_i)$;}\\ & \text{if $r+1\leqslant i\leqslant r+s, \  \zeta(m) +1\leqslant j \leqslant t$ and $2a_{ij} \mid  \deg(g_j)$;}\\[2mm]
        
        M_4 \left(\Fqr{\ell_{ij}} \right), & \text{otherwise;}
        \end{cases}  
        \]
        and
        \[
         d_{ij} = \begin{cases}
         \multirow{2}{*}{$2a_{ij}$,} & \text{if $1\leqslant i\leqslant \zeta(n)$ and $\zeta(m)+1\leqslant j\leqslant t+u$;}\\ & \text{if $\zeta(n)+1\leqslant i\leqslant r+s$ and $1\leqslant j\leqslant \zeta(m)$;}\\[2mm] 
           \multirow{3}{*}{$\dfrac{a_{ij}}{2}$,} & \text{if $\zeta(n) +1\leqslant i \leqslant r $ and $\zeta(m) +1\leqslant j \leqslant t$;} \\ & \text{if $\zeta(n) +1\leqslant i \leqslant r, \ t+1\leqslant j\leqslant t+u$ and $2a_{ij} \nmid \deg(f_i)$;}\\ & \text{if $r+1\leqslant i\leqslant r+s, \  \zeta(m) +1\leqslant j \leqslant t$ and $2a_{ij} \nmid \deg(g_j)$;}\\[2mm]
           a_{ij}, & \text{otherwise.}
        \end{cases}
        \]
        
        \item Denote $\ell_{ij} = \lcm(\deg(f_i), \deg(h_j))$ and $a_{ij} = \gcd(\deg(f_i), \deg(h_j))$  for $1\leqslant i\leqslant r+s$ and $1\leqslant j\leqslant l+k$.  Then
        \[
        \Fq[D_n \times Q_m] \cong \Fq[D_n \times D_m] \oplus \bigoplus_{i=1}^{r+s} \bigoplus_{j=1}^{l+k} \left( B_{ij}^{\oplus d_{ij}} \right),
        \]
        where
        \[
        B_{ij} = \begin{cases}
        \left(\Fq\right)^{\oplus 4},  & \text{if $1 \leqslant i \leqslant \zeta(n)$ and $1 \leqslant j \leqslant \mu(m)$;} \\[2mm]
        
        \multirow{2}{*}{$M_2 \left(\Fqr{\ell_{ij}/2} \right)$,} & \text{if $\zeta(n) +1\leqslant i \leqslant r$ and $1 \leqslant j \leqslant  \mu(m)$;} \\ & \text{if $1 \leqslant i \leqslant \zeta(n)$ and $\mu(m) +1\leqslant j \leqslant l$;} \\[2mm]
        
        \multirow{2}{*}{$M_2 \left(\Fqr{\ell_{ij}} \right)$,} & \text{if $r + 1\leqslant i \leqslant r + s$ and $1 \leqslant j \leqslant  \mu(m)$;} \\ & \text{if $1 \leqslant i \leqslant \zeta(n)$ and $l +1\leqslant j \leqslant l + k$;} \\[2mm]

        \multirow{3}{*}{$M_4 \left(\Fqr{\ell_{ij}/2} \right)$,} & \text{if $\zeta(n) +1\leqslant i \leqslant r$ and $\mu(m) +1\leqslant j \leqslant l$;} \\ & \text{if $\zeta(n) +1\leqslant i \leqslant r, \ l+1\leqslant j\leqslant l+k$ and $2a_{ij} \mid  \deg(f_i)$;}\\ & \text{if $r+1\leqslant i\leqslant r+s, \ \mu(m) +1\leqslant j \leqslant l$ and $2a_{ij} \mid  \deg(h_j)$;} \\[2mm]
        
        M_4 \left(\Fqr{\ell_{ij}} \right), & \text{otherwise;}
        \end{cases}            
        \]
        and 
         \[
         d_{ij} = \begin{cases}
         \multirow{2}{*}{$2a_{ij}$,} & \text{if $1\leqslant i\leqslant \zeta(n)$ and $\mu(m)+1\leqslant j\leqslant l+k$;}\\ & \text{if $\zeta(n)+1\leqslant i\leqslant r+s$ and $1\leqslant j\leqslant \mu(m)$;}\\[2mm] 
           \multirow{3}{*}{$\dfrac{a_{ij}}{2}$,} & \text{if $\zeta(n) +1\leqslant i \leqslant r$ and $\mu(m) +1\leqslant j \leqslant l$;} \\ & \text{if $\zeta(n) +1\leqslant i \leqslant r, \ l+1\leqslant j\leqslant l+k$ and $2a_{ij} \nmid \deg(f_i)$;}\\ & \text{if $r+1\leqslant i\leqslant r+s, \  \mu(m) +1\leqslant j \leqslant l$ and $2a_{ij} \nmid \deg(h_j)$;}\\[2mm]
           a_{ij}, & \text{otherwise.}
        \end{cases}
        \]
    \end{corolenum} 
\end{theorem}

\begin{proof} $(i)$ We will use the decompositions given in $\eqref{eq:descomp_ciclic}$ and $\eqref{eq:descomp_diedric}$, for $ \Fq[C_a]$ and  $ \Fq[D_n]$ respectively. 
Applying \cref{tpproperties} and the isomorphism given in (\ref{isom_direct_tens}),  we obtain
\begin{eqnarray*}
    \Fq[D_n \times C_a]  &\overset{\scriptscriptstyle\eqref{isom_direct_tens}}{\cong}  &\Fq[D_n] \otimes_{\Fq} \Fq[C_a]  \\    
    & \cong&  \left( \bigoplus_{i=1}^{r+s} A_i \right) \otimes_{\Fq} \left(\bigoplus_{j=1}^{b} \Fqr{\deg{(p_j)}} \right)  \\
& \overset{\scriptscriptstyle \text{Lem. \ref{tpproperties}}}{\cong} & \bigoplus_{i=1}^{r+s} \bigoplus_{j=1}^{b} \left(A_i \otimes_{\Fq}  \Fqr{\deg(p_j)}\right).  
\end{eqnarray*}

We will now study the tensor product $A_i \otimes_{\Fq}  \Fqr{\deg(p_j)}$ for the different values of $i$. 
\begin{itemize}
\item If $1 \leqslant i \leqslant \zeta(n)$, then $A_i = \Fq \oplus \Fq$, and \cref{tpfields,tpproperties:directsum} imply that

\[
A_i \otimes_{\Fq}  \Fqr{\deg(p_j)} = \left( \Fq \oplus \Fq \right)  \otimes_{\Fq}  \Fqr{\deg(p_j)} \cong \Fqr{\deg(p_j)} \oplus \Fqr{\deg(p_j)}.
\]

Therefore, $A_{ij} = \Fqr{\deg(p_j)} \oplus \Fqr{\deg(p_j)}$. Moreover, since $f_i\in\{\mx-1,\mx+1\}$, it follows that $ a_{ij}= 1 = d_{ij}$.

\item If  $\zeta(n) +1\leqslant i \leqslant r$, then $A_i =  M_2 \left(\Fqr{\deg(f_i)/2} \right)$ and \cref{tpmatrices} implies that

\[
A_i \otimes_{\Fq}  \Fqr{\deg({p_j})} = M_2 \left(\Fqr{\deg(f_i)/2} \right)  \otimes_{\Fq}  \Fqr{\deg(p_j)} \cong  M_2 \left(\Fqr{\deg(f_i)/2} \otimes_{\Fq} \Fqr{\deg(p_j)} \right).
\]
Now,  we use Lemma \ref{lem:gcd}  and \cref{tpfields} to conclude this case:

\begin{itemize}
\item If $2 a_{ij} \mid  \deg(f_i)$, then  
  \begin{alignat*}{5}
  & \lcm \left(\frac{\deg(f_i)}{2}, \deg(p_j) \right) = \dfrac{\ell_{ij}}{2}\\
  \text{and}\quad & \gcd \left(\frac{\deg(f_i)}{2}, \deg(p_j) \right) = a_{ij}.
  \end{alignat*}
Therefore  $A_{ij} = M_2 \left(\Fqr{\ell_{ij}/2} \right)$ and $d_{ij} =a_{ij}$.

\item If $2 a_{ij} \nmid \deg(f_i)$, then 
  \begin{alignat*}{5}
    & \lcm \left(\frac{\deg(f_i)}{2}, \deg(p_j) \right) = \ell_{ij} \\
   \text{and}\quad& \gcd \left(\frac{\deg(f_i)}{2}, \deg(p_j) \right)  = \frac{a_{ij}}{2}.
  \end{alignat*}
Therefore, $A_{ij} = M_2 \left(\Fqr{\ell_{ij}} \right)$ and $d_{ij} = \frac{a_{ij}}{2}$.
\end{itemize}

\item If $r+1 \leqslant i \leqslant r + s$, then $A_i =  M_2 \left(\Fqr{\deg(f_i)} \right)$ and \cref{tpfields,tpmatrices} imply that
\[
A_i \otimes_{\Fq}  \Fqr{\deg(p_j)} = M_2 \left(\Fqr{\deg(f_i)} \right)  \otimes_{\Fq}  \Fqr{\deg(p_j)} \cong  M_2 \left(\Fqr{\deg(f_i)} \otimes_{\Fq} \Fqr{\deg(p_j)} \right) \cong   M_2 \left(\Fqr{\ell_{ij}} \right)^{\oplus a_{ij}}
\]

Therefore, $A_{ij} = M_2 \left(\Fqr{\ell_{ij}} \right)$ and $d_{ij} =a_{ij}$.
\end{itemize}

\noindent We can argue in a similar way to obtain the decompositions given in statements $(ii)$ and $(iii)$.
\end{proof}

\medskip

\begin{example}
Using \cref{numberofcodes} we can compute now the number of group codes, for instance, of the group algebra $\mathbb{F}_3[D_4 \times C_5]$. Note that, by Theorem \ref{corolCxD}, the corresponding Wedderburn-Artin decomposition is $$\mathbb{F}_3[D_4 \times C_5] \cong 4\mathbb{F}_3 \oplus 4\mathbb{F}_{3^4} \oplus M_2(\mathbb{F}_3) \oplus M_2(\mathbb{F}_{3^4}).$$ Hence the number of $(D_4 \times C_5)$-codes over $\F_3$ is $2^4 \cdot 2^4 \cdot 6 \cdot 84 = 129\,024$. Analogously, since $$\mathbb{F}_3[D_4 \times D_4] \cong 16\mathbb{F}_3 \oplus 8M_2(\mathbb{F}_3) \oplus M_4(\mathbb{F}_3),$$
by Theorem \ref{corolDxD}, the number of $(D_4 \times D_4)$-codes
over $\F_3$ is $2^{16}\cdot 6^8 \cdot 212 = 23\,335\,966\,605\,312$. 
\end{example}

\section{Some constructions of new linear codes}\label{sec_examples}

The aim of this section is to show how the theoretical results stated in this paper can be applied and, in doing so, we present new linear codes that improve the best-known minimum distance for their length and dimension. Let us first illustrate Theorem \ref{theo:productes} with some concrete examples, in which we obtain the explicit Wedderburn–Artin decomposition of the group algebra of certain groups that are direct products of two groups, one of which is always dihedral. We also construct the corresponding isomorphism explicitly, allowing us to compute generators of the codes in the group algebra and to analyse their corresponding minimum distances.

\subsection{Linear codes over finite fields}

\begin{example}\label{ex:5.1}
Let $q=3$ and consider the group algebra $\mathbb{F}_3[D_5\times C_{11}]$, where $D_5=\langle x, y\,|\,  x^5 =y^2= 1, y^{-1}xy = x^{-1} \rangle$ and $C_{11}=\langle z \,|\,  z^{11} = 1\rangle$. According to Theorem \ref{corolCxD},  the Wedderburn-Artin decomposition of $\mathbb{F}_3[D_5\times C_{11}]$ depends on the factorisation of the polynomials $\mx^5-1$ and $\mx^{11}-1$ into irreducible polynomials of $\mathbb{F}_3[\mx]$.  Following Notation \ref{notation} and Theorem \ref{corolCxD},  one has $\zeta(5)=1, \, r=2, \, s=0, \, b=3$, since
\[
\begin{aligned}
\mx^5-1
&=
\underbrace{(\mx-1)}_{f_1}
\underbrace{(\mx^4+\mx^3+\mx^2+\mx+1)}_{f_2=f_2^*}\quad\text{and}
\\[1.2em]
\mx^{11}-1
&=
\underbrace{(\mx-1)}_{p_1}
\underbrace{(\mx^5-\mx^3+\mx^2-\mx-1)}_{p_2}
\underbrace{(\mx^5+\mx^4-\mx^3+\mx^2-1)}_{p_3}.
\end{aligned}
\]
\noindent The values of $a_{ij}, d_{ij}$ and $\ell_{ij}$ are organised in the following table:
\begin{center}
\begin{tblr}{colspec={cccc}, columns={mode=dmath}, vlines}
\hline
\diagbox{i}{j} & 1 & 2 & 3 \\ \hline
\SetCell[r=3]{m} 1 & a_{11} = 1 & a_{12} = 1 & a_{13} = 1 \\ 
& d_{11} = 1 & d_{12} = 1 & d_{13} = 1 \\
& \ell_{11} = 1 & \ell_{12} = 5 & \ell_{13} = 5 \\ \hline
\SetCell[r=3]{m} 2 & a_{21} = 1 & a_{22} = 1 & a_{23} = 1 \\
& d_{21} = 1 & d_{22} = 1 & d_{23} = 1 \\
& \ell_{21} = 4 & \ell_{22} = 20 & \ell_{23} = 20 \\ \hline
\end{tblr}
\end{center}
Therefore, the blocks $A_{ij}$ are as follows:
\begin{center}
\begin{tblr}{colspec={cccc}, rows={mode=dmath}, vlines, hlines}
\diagbox{i}{j}  & 1 & 2 & 3 \\
1 & \mathbb{F}_3 \oplus \mathbb{F}_3 & \mathbb{F}_{3^5} \oplus \mathbb{F}_{3^5} & \mathbb{F}_{3^5} \oplus \mathbb{F}_{3^5} \\
2 & M_2(\mathbb{F}_{3^2}) & M_2(\mathbb{F}_{3^{10}}) &  M_2(\mathbb{F}_{3^{10}}) \\
\end{tblr}
\end{center}
After a proper reordering of the summands, the information in both tables yields:
\begin{equation}\label{descomp_D5xC11}
\mathbb{F}_3[D_5\times C_{11}] \overset{\scriptscriptstyle \psi}{\cong} 2\mathbb{F}_3 \oplus 4\mathbb{F}_{3^5} \oplus M_2(\mathbb{F}_{3^2}) \oplus 2 M_2(\mathbb{F}_{3^{10}}).
\end{equation}

Now we are going to present the specific expression for the isomorphism $\psi$ through the images of the generators $x,y, z$ of the direct product group $D_5\times C_{11}$. To this end, we proceed as in the proof of \cref{Wedderburntp}. Let us first recall from \cite{Brochero2015} the definition of the isomorphism $\psi_1$ between $\F_3[D_5]$ and its Wedderburn-Artin decomposition, which is $\F_3\oplus\F_3\oplus M_2(\F_{3^2})$ (see (\ref{eq:descomp_diedric})). Let $\alpha\in \mathbb{F}_{3^2}$ be a primitive element such that $\alpha^2+\alpha-1=0$. Then:
\begin{eqnarray*}
\psi_1(x) &=& \left(1, 1, \begin{pmatrix} 0&1 \\-1 & \alpha\end{pmatrix}\right),\\
\psi_1(y) &=& \left(1, -1, \begin{pmatrix} 1&\alpha^5 \\0 & -1\end{pmatrix}\right).
\end{eqnarray*}
The isomorphism $\psi_2$ between $\mathbb{F}_3[C_{11}]$ and its decomposition $\F_3\oplus \F_{3^5}\oplus \F_{3^5}$ is a direct consequence of the Chinese Remainder Theorem (see (\ref{eq:descomp_ciclic})). In fact, if $\beta\in\F_{3^5}$ is a primitive element such that $\beta^5-\beta+1=0$, then:
$$\psi_2(z) = \left(1, \beta^{44},\beta^{22}\right).$$
Following \cref{Wedderburntp} we get $\F_3[D_{5}\times C_{11}]\overset{\scriptscriptstyle \psi}{\cong} \F_3[D_5]\otimes_{\F_3}\F_3[C_{11}]$ where $\psi=\psi_1\otimes\psi_2$. Hence, following the reordering of the summands given in (\ref{descomp_D5xC11}), we deduce:

\begin{alignat*}{5}
\psi(x) &{} =  \psi_1(x)\otimes \psi_2(1)&&{} =  \Bigg(1 , 1 , 1 , 1 , 1 , 1 , \begin{pmatrix} 0 & 1 \\ -1 & \alpha \end{pmatrix} , \begin{pmatrix} 0 & 1 \\ -1 & \alpha \end{pmatrix} , \begin{pmatrix} 0 & 1 \\ -1 & \alpha \end{pmatrix} \Bigg),\\[1.5mm]
\psi(y) &{} =  \psi_1(y)\otimes \psi_2(1)&&{} =  \Bigg( 1 , -1 , 1 , -1 , 1 , -1 , \begin{pmatrix} 1 & \alpha^5 \\ 0 & -1 \end{pmatrix} , \begin{pmatrix} 1 & \alpha^5 \\ 0 & -1 \end{pmatrix} , \begin{pmatrix} 1 & \alpha^5 \\ 0 & -1 \end{pmatrix}\Bigg), \\[1.5mm]
\psi(z) &{} = \psi_1(1)\otimes \psi_2(z)&&{} = \Bigg(1 , 1 , \beta^{44} , \beta^{44} , \beta^{22} , \beta^{22} , \begin{pmatrix} 1 & 0 \\ 0 & 1 \end{pmatrix} , \begin{pmatrix} \beta^{44} & 0 \\ 0 & \beta^{44} \end{pmatrix} , \begin{pmatrix} \beta^{22} & 0 \\ 0 & \beta^{22} \end{pmatrix}\Bigg).
\end{alignat*}

Consider a primitive element $\xi\in\mathbb{F}_{3^{10}}$ such that $\xi^{10}-\xi^6-\xi^5-\xi^4+\xi-1 = 0$. Take, for instance,
$$\cC\cong \pmb{0} \oplus \pmb{0} \oplus \pmb{0} \oplus \mathbb{F}_{3^5} \oplus \mathbb{F}_{3^5} \oplus \pmb{0} \oplus \left\langle \begin{pmatrix} 1 & \alpha^3 \\ 0 & 0 \end{pmatrix} \right\rangle_{M_2\left(\mathbb{F}_{3^2}\right)} \oplus \left\langle \begin{pmatrix} 1 & \xi^{9728} \\ 0 & 0 \end{pmatrix} \right\rangle_{M_2\left(\mathbb{F}_{3^{10}}\right)} \oplus \pmb{0}.$$
Its dimension can be easily computed utilising \cref{code_dim}. Indeed, since $\dim_{\mathbb{F}_3}\left(\F_{3^5}\right)=5$, and
\begin{small}
\begin{gather*}
\dim_{\mathbb{F}_3}\left(\left\langle \begin{pmatrix} 1 & \alpha^3 \\ 0 & 0 \end{pmatrix} \right\rangle_{M_2\left(\mathbb{F}_{3^2}\right)}\right) = 2\cdot 2\cdot 1=4, \qquad
\dim_{\mathbb{F}_3}\left(\left\langle \begin{pmatrix} 1 & \xi^{9728} \\ 0 & 0 \end{pmatrix} \right\rangle_{M_2\left(\mathbb{F}_{3^{10}}\right)}\right) = 2\cdot 10\cdot 1=20,
\end{gather*} 
\end{small}%
we can conclude that the dimension of $\mathcal{C}$ is $5+5+4+20=34$. Using the software \textsc{Magma} \cite{Magma}, we were able to construct this concrete $\psi$ and its inverse, and so we computed a generator of $\mathcal{C}$ in $\mathbb{F}_3[D_5\times C_{11}]$, which is 
\begin{dmath*}
x^3z^{10} + x^2yz^{10} + x^2yz^9 + x^4yz^{10} - x^3z^8 - x^2yz^8 - x^3yz^8 - x^2z^9 + xyz^9 + xz^{10} - x^3z^7 - x^2yz^7 - x^3yz^7 + x^4z^7 - x^2z^8 + x^4yz^8 - xz^9 - yz^9 + x^3yz^6 - x^4z^6 - x^2z^7 + xyz^7 + xz^8 + z^9 + x^2yz^5 - x^4z^5 + x^2z^6 - xyz^6 - x^4yz^6 - x^3z^4 + x^3yz^4 - x^4z^4 - x^2z^5 + x^4yz^5 - xz^6 + x^2yz^3 - x^4z^3 + x^4yz^4 - yz^5 + x^3yz^2 + x^4z^2 - x^2z^3 + x^4yz^3 - yz^4 + z^5 + x^2yz - x^4z + x^2z^2 + xyz^2 - yz^3 + z^4 - x^3 - x^4 - x^2z + x^4yz + z^3 - x^2 - yz - x + z - 1.
 \end{dmath*}
Moreover, we determined that $\cC$ has minimum distance $35$. This improves upon the best known minimum distance for ternary linear codes with parameters $[110,34]_3$, as recorded in Grassl’s tables \cite{Grasslcodetables}. \hfill $\blacksquare$
\end{example}

Below we present other examples of group codes over $\mathbb{F}_3$ or $\mathbb{F}_5$ associated with the direct product of a dihedral group and a cyclic group. All of them improve the best-known minimum distance for their length and dimension (\emph{cf.} \cite{Grasslcodetables}). To obtain the generator element of each one of these codes we have proceeded as in Example \ref{ex:5.1}, that is, constructing the specific Wedderburn-Artin decomposition according to Theorem \ref{corolCxD} and the corresponding isomorphism provided by \cref{Wedderburntp}. The results were obtained using \textsc{Magma} by random search.

Let $D_n=\langle x, y\,|\, x^{n} =y^2= 1, y^{-1}xy = x^{-1} \rangle$
and $C_t=\langle z\,|\, z^t = 1\rangle$. The following table shows
some of the new codes we computed.

\def\generator#1{\parbox{12.5cm}{\begin{dmath*}#1\end{dmath*}}}

\begin{small}
\renewcommand{\arraystretch}{1.2}
\def\eqindent{I}
\begin{longtable}{|c|c|c|c|c|}
\caption{New linear codes obtained from certain group codes.}
\label{tablacodigoss}\\
\hline
Generator of an ideal of $\mathbb{F}_q[D_n\times C_t]$ & $q$ & $n$ & $t$ & Parameters \\ \hline
\endfirsthead
\hline
\hfill Generator of an ideal of $\mathbb{F}_q[D_n\times C_t]$\hfill\ & $q$ & $n$ & $t$ & Parameters \\ \hline
\endhead
\multicolumn{5}{|r|}{\emph{continued on the next page}}\\
\hline
\endfoot
\endlastfoot
	\generator{-x^7yz^3 - x^6yz^2 + x^7yz^2 - x^8yz^3 - x^6yz - x^7yz - x^8yz^2 - x^4yz^3 + x^9yz^3 - x^{10}z^3 - x^7 - x^6y - x^7y - x^8 - x^5yz + x^8yz + x^4yz^2 - x^9yz^2 + x^{10}z^2 - x^4z^3 - x^3yz^3 - x^6 + x^5y + x^9 + x^4yz - x^9yz - x^{10}z + x^4z^2 + x^3yz^2 - x^{10}yz^2 - x^{12}z^3 - x^5 + x^4y + x^{10} - x^4z - x^3yz + x^{11}yz^2 + x^{12}z^2 - xyz^3 - x^{12}yz^3 + x^4 - x^3y - x^{11} - x^2yz + x^{11}yz - x^{11}z^3 + xyz^2 + x^3 + x^2y - x^{11}y + x^{12} + xyz + yz^2 - z^3 - x^2 + xy + x^{12}y + z^2 + x + y - z - 1} & 3 &13 & 4 & $[104, 55, 19]_3$ \\\hline
	\generator{x^3z^{10} + x^3yz^{10} + x^3z^9 + x^3yz^9 - x^2z^{10} - xyz^{10} - x^4yz^{10} - x^3z^8 + x^2yz^8 - x^4z^8 - x^2z^9 - xyz^9 - x^4yz^9 + yz^{10} + x^3z^7 + x^3yz^7 + x^2z^8 + xyz^8 + yz^9 - z^{10} - x^3z^6 + x^2yz^6 - x^4z^6 - x^2z^7 - xyz^7 - x^4yz^7 + xz^8 + yz^8 - z^9 - x^2yz^5 - x^3yz^5 + x^2z^6 + xyz^6 + yz^7 - z^8 - x^3z^4 + x^2yz^4 - x^4z^4 - xyz^5 - x^4yz^5 + xz^6 + yz^6 - z^7 - x^3z^3 + x^2yz^3 - x^4z^3 + x^2z^4 + xyz^4 + yz^5 - z^6 - x^2yz^2 - x^3yz^2 + x^2z^3 + xyz^3 + xz^4 + yz^4 - z^5 - x^2yz - x^3yz - x^4z - xyz^2 - x^4yz^2 + xz^3 + yz^3 - z^4 - x^2y + x^4 + xyz + yz^2 - z^3 + x^2 - x^4y + xz + yz - z^2 - y - z} & 3 &5 & 11 & $[110, 42, 29]_3$ \\\hline
	\generator{-x^6z^4 - x^7z^4 - x^6z^3 + x^6yz^3 - x^7z^3 + x^4yz^4 - x^8z^4 - x^6z^2 + x^5yz^2 + x^6yz^2 - x^7z^2 + x^4yz^3 - x^7yz^3 - x^8z^3 - x^8yz^4 - x^6z + x^5yz + x^6yz - x^7z - x^4yz^2 - x^7yz^2 - x^8z^2 + x^8yz^3 - x^2yz^4 + x^9yz^4 - x^{10}z^4 - x^6 - x^5y + x^6y - x^7 + x^4yz + x^7yz - x^8z - x^3yz^2 - x^8yz^2 - x^2yz^3 - x^{10}z^3 - x^2z^4 + x^5 - x^4y - x^7y - x^8 + x^8yz - x^2yz^2 - x^{10}z^2 - x^2z^3 + xyz^3 + x^{10}yz^3 + yz^4 + x^4 - x^3y + x^8y + x^9 - x^9yz - x^{10}z - x^2z^2 + yz^3 + x^3 + x^2y + x^9y - x^{10} - x^2z - xyz + x^{10}yz + yz^2 - x^2 + xy - x^{10}y + x + y + 1} & 3 &11 & 5 & $[110, 41, 30]_3$ \\\hline
	\generator{x^{6}z^{4} - x^{6}yz^{4} + x^{7}z^{4} + x^{6}z^{3} + x^{5}yz^{3} + x^{7}z^{3} - x^{4}yz^{4} - x^{7}yz^{4} + x^{8}z^{4} + x^{6}z^{2} - x^{5}yz^{2} - x^{6}yz^{2} + x^{7}z^{2} + x^{8}z^{3} - x^{3}yz^{4} - x^{8}yz^{4} + x^{6}z + x^{5}yz - x^{6}yz + x^{7}z - x^{4}yz^{2} - x^{7}yz^{2} + x^{8}z^{2} - x^{3}yz^{3} + x^{8}yz^{3} + x^{10}z^{4} - x^{5}y - x^{4}yz + x^{7}yz + x^{8}z + x^{3}yz^{2} - x^{8}yz^{2} - x^{2}yz^{3} - x^{9}yz^{3} + x^{10}z^{3} + x^{2}z^{4} - xyz^{4} + x^{7}y + x^{3}yz + x^{8}yz + x^{2}yz^{2} + x^{10}z^{2} + x^{2}z^{3} + x^{2}yz + x^{10}z + x^{2}z^{2} - x^{10}yz^{2} + yz^{3} - z^{4} - x^{2}y + x^{9}y + x^{2}z - xyz + x^{10}yz + yz^{2} - z^{3} - xy - z^{2} + y - z - 1} & 3 &11 & 5 & $[110, 58, 21]_3$ \\\hline
	\generator{-x^5yz^4 + x^6yz^4 - x^6yz^3 - x^5z^4 + x^4yz^4 + x^7yz^4 + x^5yz^2 + x^6yz^2 - x^5z^3 - x^4yz^3 - x^7yz^3 - x^4z^4 + x^3yz^4 + x^8yz^4 - x^9z^4 + x^6yz - x^5z^2 + x^4yz^2 + x^7yz^2 - x^4z^3 + x^3yz^3 - x^9z^3 - x^3z^4 - x^2yz^4 - x^9yz^4 - x^6 + x^5y - x^6y - x^7 - x^5z + x^4yz - x^4z^2 + x^8yz^2 - x^9z^2 - x^3z^3 + x^2yz^3 + x^9yz^3 + xyz^4 - x^{10}yz^4 - x^5 - x^4y - x^8 - x^4z - x^9z - x^3z^2 - x^9yz^2 - xyz^3 - x^{10}yz^3 - xz^4 - yz^4 - x^4 - x^3y - x^8y - x^9 - x^3z - x^9yz - xyz^2 + x^{10}yz^2 - xz^3 + z^4 - x^3 + x^2y - x^{10} + xyz - xz^2 + z^3 - x^2 + xy - x^{10}y - xz - yz + z^2 - x + z + 1} & 3 &11 & 5 & $[110, 59, 20]_3$ \\\hline
	\generator{-x^{3}z^{12} + x^{2}yz^{12} - x^{4}z^{12} - x^{2}yz^{11} + x^{2}z^{12} + xyz^{12} - x^{3}z^{10} + x^{2}yz^{10} + x^{2}z^{11} - xyz^{11} + xz^{12} - yz^{12} + x^{2}yz^{9} + x^{3}yz^{9} - x^{4}z^{9} + x^{2}z^{10} - xyz^{10} - x^{4}yz^{10} + xz^{11} + yz^{11} - x^{3}z^{8} + x^{2}yz^{8} + x^{3}yz^{8} + x^{4}z^{8} - x^{4}yz^{9} - yz^{10} - z^{11} - x^{3}z^{7} + x^{2}yz^{7} + x^{3}yz^{7} - x^{2}z^{8} - x^{4}yz^{8} + xz^{9} - x^{3}z^{6} - x^{2}yz^{6} + x^{3}yz^{6} - x^{4}z^{6} - x^{2}z^{7} - xyz^{7} + yz^{8} - z^{9} + x^{4}z^{5} + x^{2}z^{6} - xyz^{6} + x^{4}yz^{6} + xz^{7} + yz^{7} - z^{8} - x^{2}yz^{4} - x^{3}yz^{4} - x^{4}z^{4} + xyz^{5} - x^{4}yz^{5} + xz^{6} + yz^{6} - z^{7} - x^{3}z^{3} - x^{3}yz^{3} - x^{4}z^{3} + xyz^{4} - xz^{5} + yz^{5} + z^{6} + x^{3}z^{2} + x^{2}yz^{2} - x^{3}yz^{2} + x^{4}z^{2} + x^{2}z^{3} - x^{4}yz^{3} + xz^{4} - yz^{4} + z^{5} - x^{3}z - x^{3}yz - x^{4}z - x^{2}z^{2} + xyz^{2} - x^{4}yz^{2} + xz^{3} + x^{3} + x^{2}y + x^{3}y + x^{4} + x^{2}z - x^{4}yz - xz^{2} + yz^{2} - z^{3} + x^{2} + xy + x^{4}y + xz + z^{2} + x + y - z + 1} & 3 &5 & 13 & $[130, 43, 38]_3$ \\\hline
	\generator{-x^{3}yz^{12} + x^{3}z^{11} + x^{2}yz^{11} - x^{3}yz^{11} + x^{4}z^{11} - x^{2}z^{12} - x^{4}yz^{12} - x^{3}yz^{10} + x^{4}z^{10} - x^{2}z^{11} + xyz^{11} - x^{4}yz^{11} - xz^{12} + yz^{12} + x^{3}z^{9} - x^{2}z^{10} + xyz^{10} + x^{4}yz^{10} - xz^{11} + z^{12} - x^{4}z^{8} + x^{2}z^{9} - xyz^{9} + x^{4}yz^{9} + xz^{10} + yz^{10} + x^{4}z^{7} - xyz^{8} + x^{4}yz^{8} - xz^{9} - yz^{9} + z^{10} + x^{2}yz^{6} + xyz^{7} - x^{4}yz^{7} + xz^{8} + x^{3}z^{5} - x^{2}yz^{5} - x^{3}yz^{5} - x^{4}z^{5} - x^{2}z^{6} + xyz^{6} - xz^{7} + x^{3}z^{4} + x^{2}yz^{4} + x^{3}yz^{4} + x^{2}z^{5} + x^{4}yz^{5} - xz^{6} - x^{2}yz^{3} + x^{3}yz^{3} + x^{2}z^{4} - x^{4}yz^{4} - z^{6} - x^{3}z^{2} + x^{3}yz^{2} - x^{4}z^{2} - x^{2}z^{3} - xyz^{3} + x^{4}yz^{3} - xz^{4} + yz^{4} - z^{5} - x^{2}yz + x^{3}yz + x^{4}yz^{2} - xz^{3} - yz^{3} + z^{4} - x^{3} - x^{4} - x^{2}z - xyz + x^{4}yz - x^{2} - xz - yz - z^{2} - x - 1} & 3 &5 & 13 & $[130, 44, 37]_3$ \\\hline
	\generator{x^7z^2 + x^6yz^2 + 2x^7yz^2 + 2x^8z^2 + 4x^6yz + 4x^7yz + 2x^6z^2 + x^5yz^2 + 2x^8yz^2 + 2x^9z^2 + x^6y + 3x^7y + 4x^8 + 3x^5yz + x^8yz + 3x^9z + 3x^5z^2 + x^4yz^2 + 2x^9yz^2 + 3x^{10}z^2 + 2x^6 + 4x^8y + 2x^9 + 3x^5z + 3x^4yz + 3x^4z^2 + 2x^3yz^2 + 2x^{11}z^2 + 3x^5 + 3x^4y + x^9y + 3x^{10} + x^3yz + 2x^{10}yz + x^{11}z + x^3z^2 + 3x^2yz^2 + 3x^{12}z^2 + 4x^3y + 3x^{10}y + x^{11} + x^3z + 2x^{11}yz + 4x^{12}z + 4x^2z^2 + xyz^2 + 3x^{12}yz^2 + 2x^3 + 2x^2y + 3x^2z + 3xyz + 3yz^2 + 4x^2 + xy + 3x^{12}y + xz + 3yz + 3z^2 + 2x + y + 3z + 2} & 5 &13 & 3 & $[78, 41, 19]_5$ \\\hline
\end{longtable}
\end{small}

\begin{example}
Now we present the construction of another new linear code in terms of an ideal in $\mathbb{F}_5[D_{11}\times D_{2}]$. Note that $D_{2}$ is actually the Klein $4$-group, which has not either a cyclic or a \emph{pure} dihedral structure; nevertheless, the decomposition of $\mathbb{F}_5[D_2]$ may still be provided by either (\ref{eq:descomp_diedric}) or (\ref{eq:descomp_abelia}), so \cref{Wedderburntp} can be applied too. Indeed, the Wedderburn-Artin decomposition, according to Theorem \ref{corolDxD}, is as follows.
\begin{equation}\label{descomp_D11xD2}
\mathbb{F}_5[D_{11}\times D_{2}] \cong 8\mathbb{F}_5 \oplus 4 M_2(\mathbb{F}_{5^{5}}).
\end{equation}
Consider a primitive element $\alpha\in\F_{5^5}$ such that $\alpha^5-\alpha+3=0$. Let us take the code $\cal C$ that corresponds to the next ideal of $\mathbb{F}_5[D_{11}\times D_{2}]$:
\begin{small}
$$\F_5 \oplus  \F_5 \oplus\F_5\oplus \boldsymbol{0} \oplus\F_5\oplus \boldsymbol{0} \oplus\F_5\oplus \boldsymbol{0} \oplus  \left\langle\begin{pmatrix} 1&\alpha^{921}\\ 0&0\end{pmatrix}\right\rangle \oplus \left\langle\begin{pmatrix} 1&\alpha^{1749}\\ 0&0\end{pmatrix}\right\rangle \oplus\left\langle\begin{pmatrix} 1&\alpha^{1679}\\ 0&0\end{pmatrix}\right\rangle \oplus\left\langle\begin{pmatrix} 1&\alpha^{2363}\\ 0&0\end{pmatrix}\right\rangle.$$
\end{small}

Denote $D_2=\langle z, w\,|\,  z^2 =w^2= 1, wz = zw \rangle$ and $D_{11}=\langle x, y\,|\,  x^{11} =y^2= 1, yx=x^{-1}y \rangle$. Then the group code ${\cal C}$ is generated in $\mathbb{F}_5[D_{11}\times D_{2}]$ by the following element:
\begin{dmath*}3x^6zw + 2x^5yzw + 3x^6yzw + 3x^7zw + 2x^6z + 3x^5yz + 3x^6yz + 2x^7z + x^6w + x^5yw + 4x^6yw + x^7w + 3x^5zw + 3x^7yzw + 3x^8zw + 3x^6 + x^5y + 3x^6y + 3x^7 + 2x^5z + x^4yz + x^7yz + 2x^8z + x^5w + 3x^4yw + 4x^7yw + x^8w + 3x^4zw + 3x^9zw + x^5 + 3x^7y + 3x^8 + 2x^4z + 2x^8yz + 2x^9z + x^4w + 3x^3yw + 2x^8yw + x^9w + 3x^3zw + 2x^2yzw + 3x^{10}zw + x^4 + 4x^3y + 3x^8y + x^9 + 2x^3z + 2x^2yz + x^9yz + 2x^{10}z + x^3w + 3x^9yw + x^{10}w + 3x^2zw + xyzw + x^{10}yzw + x^3 + 2x^2y + 4x^9y + 3x^{10} + 2x^2z + xyz + x^2w + 2xyw + x^{10}yw + 3xzw + yzw + 3x^2 + 3xy + 2xz + 3yz + xw + 3zw + x + 4y + 2z + w.
\end{dmath*}
It can be checked that this code has parameters $[88,45,21]_5$, so it has better minimum distance than any other known $5$-ary code with these length and dimension, according to \cite{Grasslcodetables}. \hfill $\blacksquare$
\end{example}

We now illustrate another important application of the Wedderburn–Artin decomposition of a group code: the construction of subcodes and supercodes (that is, codes contained in the original code or containing it). Starting from a given code, we can systematically search the corresponding lattice to provide new codes with improved parameters, using methods such as the so-called Construction X, which we recall below.

\begin{proposition}[{Construction X, \cite[Ch. 18, §7.1]{MS}}]
Given a code $\cC_1$ with parameters $[n, k_1, d_1]_q$, a proper subcode $\cC_2 \subseteq \cC_1$ with parameters $[n, k_2, d_2]_q$, and an auxiliary code $\cC_3$ with parameters $[n', k_1 - k_2, d_3]_q$, then there is a code $\cC$ with parameters $[n + n', k_1, d]_q$ with $d \geq \min(d_1 + d_3, d_2)$.
\end{proposition}

\begin{example}
Let $\cC_1$ be the $(D_{11}\times D_{2})$-code constructed in the previous example, so its parameters are $[88,45,21]_5$. Consider the following subcode of $\cC_1$:
\begin{eqnarray*}
\cC_2&\cong& \F_5 \oplus  \F_5 \oplus\boldsymbol{0}\oplus \boldsymbol{0} \oplus\F_5\oplus \boldsymbol{0} \oplus\F_5\oplus \boldsymbol{0} \oplus\\
&\oplus& \left\langle\begin{pmatrix} 1&\alpha^{921}\\ 0&0\end{pmatrix}\right\rangle \oplus \left\langle\begin{pmatrix} 1&\alpha^{1749}\\ 0&0\end{pmatrix}\right\rangle \oplus\left\langle\begin{pmatrix} 1&\alpha^{1679}\\ 0&0\end{pmatrix}\right\rangle \oplus\left\langle\begin{pmatrix} 1&\alpha^{2363}\\ 0&0\end{pmatrix}\right\rangle,
\end{eqnarray*}
whose generator in the group algebra $\F_5[D_{11}\times D_{2}]$ is 
\begin{dmath*}
x^6zw + 4x^5yzw + x^7zw + 4x^6w + 3x^5yw + x^6yw + 4x^7w + x^5zw + 2x^4yzw + x^8zw + x^6 + 3x^5y + x^7 + 3x^4yz + 3x^7yz + 4x^5w + x^7yw + 4x^8w + x^4zw + 2x^3yzw + 2x^8yzw + x^9zw + 4x^5 + 2x^4y + x^8 + 2x^3yz + 4x^8yz + 4x^4w + 4x^8yw + 4x^9w + x^3zw + 4x^2yzw + 2x^9yzw + x^{10}zw + 4x^4 + x^3y + 4x^9 + 4x^2yz + 3x^9yz + 4x^3w + 2x^2yw + 4x^{10}w + x^2zw + 3xyzw + 3x^{10}yzw + 4x^3 + 4x^2y + x^9y + x^{10} + 3xyz + 2x^{10}yz + 4x^2w + 4xyw + 3x^{10}yw + xzw + 3yzw + x^2 + 2x^{10}y + 4xw + 2yw + zw + 4x + y + 4w + 3.
\end{dmath*}
Certainly $\cC_2$ has the same length as $\cC_1$, and the difference of the dimensions equals $1$. Using \textsc{Magma} we computed its minimum distance, which is $22$, so the parameters $[88,44,22]_5$ of this $\cC_2$ are equal to the best-known ones (according to \cite{Grasslcodetables}). Take as $\cC_3$ the trivial $\F_5$-code of length and dimension $1$. Therefore, applying Construction X to these nested group codes, we know that there exists a linear code with parameters $[89, 45, 22]_5$, which improves upon the best-known minimum distance for $5$-ary linear codes of length $89$ and dimension $45$ (\emph{cf.} \cite{Grasslcodetables}). \hfill $\blacksquare$
\end{example}

\begin{example}
Another illustration of how Construction X may be appplied in our setting is the next one. The codes $\cC_1$ and $\cC_2$ in Table \ref{tablacodigoss} with parameters $[130,44,37]_3$ and $[130,43,38]_3$, respectively, satisfy that $\cC_2\subseteq \cC_1$. In fact, the Wedderburn-Artin decomposition of $\F_3[D_5\times C_{13}]$ is $2\F_3 \oplus 8\F_{3^3} \oplus M_2(\F_{3^2}) \oplus 4M_2(\F_{3^6})$, and if $\xi\in\F_{3^6}$ is a primitive element such that $\xi^6-\xi^4+\xi^2-\xi-1=0$, then the decomposition of $\cC_1$ and $\cC_2$ are as follows:
\begin{eqnarray*}
 \cC_1 &\cong& \F_3\oplus \F_3\oplus\F_{3^3}\oplus \boldsymbol{0}\oplus\boldsymbol{0}\oplus \boldsymbol{0}\oplus\boldsymbol{0}\oplus \F_{3^3}\oplus\boldsymbol{0}\oplus \boldsymbol{0} \oplus\\
&&\quad \boldsymbol{0}\oplus \left\langle \begin{pmatrix} 1 & \xi^{152}\\ 0 & 0\end{pmatrix}\right\rangle \oplus \left\langle \begin{pmatrix} 1 & \xi^{597}\\ 0 & 0\end{pmatrix}\right\rangle \oplus \boldsymbol{0}\oplus  \left\langle \begin{pmatrix} 1 & \xi^{695} \\ 0 & 0\end{pmatrix}\right\rangle
\end{eqnarray*}
\medskip
\begin{eqnarray*}
\cC_2 &\cong& \F_3\oplus \boldsymbol{0}\oplus\F_{3^3}\oplus \boldsymbol{0}\oplus\boldsymbol{0}\oplus \boldsymbol{0}\oplus\boldsymbol{0}\oplus \F_{3^3}\oplus\boldsymbol{0}\oplus \boldsymbol{0} \oplus \\
&& \quad  \boldsymbol{0}\oplus \left\langle \begin{pmatrix} 1 & \xi^{152}\\ 0 & 0\end{pmatrix}\right\rangle \oplus \left\langle \begin{pmatrix} 1 & \xi^{597}\\ 0 & 0\end{pmatrix}\right\rangle \oplus \boldsymbol{0}\oplus  \left\langle \begin{pmatrix} 1 & \xi^{695} \\ 0 & 0\end{pmatrix}\right\rangle
\end{eqnarray*}
Take as $\cC_3$ the trivial $\F_3$-code of length and dimension $1$. Hence Construction X yields a linear code with parameters $[131, 44, 38]_3$, which is better than any of the current ternary linear codes of length $131$ and dimension $44$, according to \cite{Grasslcodetables}. \hfill $\blacksquare$
\end{example}

\begin{example}
  The codes $\cC_1$ and $\cC_2$ in Table \ref{tablacodigoss} with
  parameters $[110,59,20]_3$ and $[110,58,21]_3$, respectively,
  satisfy that $\cC_2\subseteq \cC_1$. Using the trivial code $\cC_3$
  with parameters $[1,1,1]_3$, we obtain a code with parameters
  $[111,59,21]_3$ which also improves the current lower bound from
  \cite{Grasslcodetables}.\hfill $\blacksquare$
\end{example}

\subsection{Application of group codes to quantum error-correcting codes}

Within quantum information processing, the ability of controlling quantum errors is fundamental in order to construct efficient quantum computers; in fact, the research area of quantum error control is growing tremendously in the last few years. Among the possible quantum error-correcting codes, there is a particular type of them which can be designed via classical linear codes, the so-called CSS quantum codes, whose (binary) construction is due to Calderbank \& Shor \cite{calderbank1996good} and Steane \cite{steane1996simple}. Later on, this was generalised to non-binary alphabets (\emph{cf.} \cite{CRSS, ketkar2006nonbinary}). These classical codes are not arbitrary: there must verify certain orthogonality condition. In the current literature, this condition is sometimes checked using brute-force algorithms (see, for instance, \cite{Yu24}). However, group codes can play an important role in this framework, since a systematic method can be easily designed to construct classical codes that meet the required orthogonality condition. In this section our objective is to illustrate such a method, and for that goal we first collect some basic concepts concerning CSS quantum codes; we refer readers interested in additional details to the book \cite{Lidar2013}.

Recall that, if $x=(x_1,x_2,\ldots,x_n)$ and $y=(y_1,y_2,\ldots,y_n)$ are two elements in $\mathbb{F}_q^n$, then their (Euclidean) inner product is defined as $x\cdot y=\sum_{i=1}^n x_i y_i$.  In this way,  given a linear code $\mathcal{C}$ of $\Fq^n$,   the (Euclidean) {\em dual code} of $\mathcal{C}$ is $\mathcal{C}^{\perp}=\{y\in\mathbb{F}_q^n\colon\; x\cdot y =0, \;\forall \,  x\in\mathcal{C}\}$. 

Let $\mathbb{C}$ be the complex field and $V_n=\mathbb{C}^q\otimes \cdots \otimes \mathbb{C}^q=(\mathbb{C}^q)^{\otimes n}$ be the Hilbert space of dimension $q^n$. The CSS construction works as follows: given two linear codes $\mathcal{C}_1, \mathcal{C}_2\subseteq \mathbb{F}_q^n$ of dimensions $k_1$ and $k_2$ respectively, and such that  $\mathcal{C}_2^{\perp}\subseteq \mathcal{C}_1$,  it is possible to construct a \emph{quantum error-correcting code} (QECC) from them, denoted $\text{CSS}(\mathcal{C}_1, \mathcal{C}_2)$, which will be a subspace of $V_n$ of dimension $q^{k_1+k_2-n} $. Moreover, the minimum distance of $\text{CSS}(\mathcal{C}_1, \mathcal{C}_2)$ is $d$ whenever all errors acting on at most $d-1$ of the $n$ subsystems (tensor components) can be detected or act trivially on the code. We say then that  $\text{CSS}(\mathcal{C}_1, \mathcal{C}_2)$ is an $[\![n,k_1+k_2-n,d]\!]_q$ quantum code. 

\begin{theorem}
\label{quantum_CSS}
Given linear codes $\mathcal{C}_1$ and $\mathcal{C}_2$ over $\mathbb{F}_q$ with parameters $[n,k_1,d_1]_q$ and $[n,k_2,d_2]_q$, respectively, such that $\mathcal{C}_2^{\perp}\subseteq\mathcal{C}_1$, there exists a quantum error-correcting code $\mathcal{Q}=\operatorname{CSS}(\mathcal{C}_1, \mathcal{C}_2)$ with parameters $[\![n,k_1+k_2-n,d_{\mathcal{Q}}]\!]_q$, where $d_{\mathcal{Q}}$ is the minimum of the weights of codewords lying in $(\mathcal{C}_1\setminus\mathcal{C}_2^{\perp})\cup (\mathcal{C}_2\setminus\mathcal{C}_1^{\perp})$ whenever $\mathcal{C}_2^{\perp}\subsetneq \mathcal{C}_1$, and otherwise $d_{\mathcal{Q}}=\min(d_1,d_2)$.
\end{theorem}

In the next two examples we illustrate how the Wedderburn-Artin decomposition of a group algebra (Theorem \ref{Wedderburn}) can be used to provide a systematic construction of classical codes that verify the required orthogonality relation in Theorem \ref{quantum_CSS}, providing in this way a method to obtain quantum error-correcting group codes via this CSS construction. To this end, not only must the Wedderburn–Artin decomposition of the group algebra be known, but we must also determine the duals of its ideals. In the case of the dihedral group algebra, Vedenev and Deundyak did it in \cite[Theorem 5]{Vedenev2021}.

\begin{example}
Let us consider $q=3$ and let us work on the group algebra $\mathbb{F}_3[D_{16}]$. Following the notation in \cref{group-algebra-decomposition}, it is clear that $\zeta(16)=2$, and it can be checked that $$\mx^{16}-1=(\mx-1)(\mx+1)(\mx^2+1)(\mx^2+\mx-1)(\mx^2-\mx-1)(\mx^4+\mx^2-1)(\mx^4-\mx^2-1)=f_1f_2f_3f_4f_4^*f_5f_5^*,$$ so $r=3$ and $s=2$. Consider $\beta$ a primitive element of $\mathbb{F}_{3^4}$ and take the $D_{16}$-code $$\cC\cong(\mathbb{F}_3\oplus \mathbb{F}_3)\oplus(\pmb{0}\oplus \mathbb{F}_3)\oplus  \left\langle \begin{pmatrix}
1&0\\0&1 \end{pmatrix} \right\rangle_{M_2(\mathbb{F}_{3})}\oplus  \left\langle \begin{pmatrix}
1&0\\0&0 \end{pmatrix} \right\rangle_{M_2(\mathbb{F}_{3^2})}\oplus  \left\langle \begin{pmatrix}
1&-\beta\\0&0 \end{pmatrix} \right\rangle_{M_2(\mathbb{F}_{3^4})}.$$ 
Note that $\mathcal{C}$ is a $D_{16}$-code of dimension $19$ (by \cref{code_dim}), and in virtue of \cite[Theorem 5]{Vedenev2021} we get that 
$$
\mathcal{C}^{\perp}\cong (\pmb{0} \oplus \pmb{0}) \oplus (\mathbb{F}_3 \oplus \pmb{0}) \oplus \left\langle \begin{pmatrix}
0&0\\0&0 \end{pmatrix} \right\rangle_{M_2(\mathbb{F}_{3})}\oplus  \left\langle \begin{pmatrix}
1&0\\0&0 \end{pmatrix} \right\rangle_{M_2(\mathbb{F}_{3^2})}\oplus  \left\langle \begin{pmatrix}
1&\beta\\0&0 \end{pmatrix} \right\rangle_{M_2(\mathbb{F}_{3^4})}.
$$ 
Besides, for each $D_{16}$-code $\mathcal{D}$ such that $\mathcal{C}^{\perp}\subseteq \mathcal{D}$, we can construct a dihedral quantum code applying the CSS construction (\emph{cf.} \cref{quantum_CSS}).   We consider for instance the supercode:
$$
\mathcal{D}\cong (\mathbb{F}_3\oplus \mathbb{F}_3) \oplus (\mathbb{F}_3 \oplus \pmb{0}) \oplus \left\langle \begin{pmatrix}
1&0\\0&1 \end{pmatrix} \right\rangle_{M_2(\mathbb{F}_{3})}\oplus  \left\langle \begin{pmatrix}
1&0\\0&1 \end{pmatrix} \right\rangle_{M_2(\mathbb{F}_{3^2})}\oplus  \left\langle \begin{pmatrix}
1&\beta\\0&0 \end{pmatrix} \right\rangle_{M_2(\mathbb{F}_{3^4})}.
$$ It holds that $\mathcal{D}$ is a $D_{16}$-code of dimension $23$. By  \cref{quantum_CSS} we obtain that ${\cal Q}_1=\operatorname{CSS}(\mathcal{D}, \mathcal{C})$ is a quantum $D_{16}$-code of length $32$ and dimension $3^{10}$. Using \textsc{Magma} \cite{Magma} we also obtained its minimum distance which is $4$, so its parameters are $[\![32,10,4]\!]_{3}$. 
\end{example}

\setlength{\arraycolsep}{4pt}

\begin{example}
Now we consider the group algebra $\mathbb{F}_3[D_{20}]$ and $q=3$. Thus $\zeta(20)=2$, and since \begin{eqnarray*}\mx^{20}-1&=&(\mx-1)(\mx+1)(\mx^2+1)(\mx^4+\mx^3+\mx^2+\mx+1) \\ & & \quad (\mx^4-\mx^3+\mx^2-\mx+1)(\mx^4+\mx^3-\mx+1)(\mx^4-\mx^3+\mx+1) \\ &=&f_1f_2f_3f_4f_5f_6f_6^*,\end{eqnarray*} 
we deduce that $r=5$ and $s=1$.  For $i=4,5$, let $\alpha_i$ be a root of $f_i$.  By 
\cite[Remark 3.2]{Brochero2015} there exists a polynomial $h_i$ in $\mathbb{F}_3[\mx]$ having degree $2$ and  $\alpha_i+\alpha_i^{-1}$ as a root. Thus, we can take 
\[
\lambda_4=-2/(\alpha_4+\alpha_4^{-1})\in \mathbb{F}_3(\alpha_4+\alpha_4^{-1})\cong  \mathbb{F}_{3^2}
\] 
and
\[
\lambda_5=-(\alpha_5+\alpha_5^{-1})/2\in \mathbb{F}_3(\alpha_5+\alpha_5^{-1})\cong  \mathbb{F}_{3^2}.
\] 
Let us take the $D_{20}$-code 
\begin{alignat*}{5}
\mathcal{C}\cong(\mathbb{F}_3\oplus \mathbb{F}_3)\oplus(\pmb{0}\oplus \mathbb{F}_3)&{}\oplus  \left\langle \begin{pmatrix}
1&0\\0&1 \end{pmatrix} \right\rangle_{M_2(\mathbb{F}_{3})}&&{}\oplus  \left\langle \begin{pmatrix}
1&\lambda_4\\0&0 \end{pmatrix} \right\rangle_{M_2(\mathbb{F}_{3}(\alpha_4+\alpha_4^{-1}))}\\
&{}\oplus  \left\langle \begin{pmatrix}
1&\lambda_5\\0&0 \end{pmatrix} \right\rangle_{M_2(\mathbb{F}_{3}(\alpha_5+\alpha_5^{-1}))}&&{}\oplus \left\langle \begin{pmatrix}
1&0\\0&0 \end{pmatrix} \right\rangle_{M_2(\mathbb{F}_{3^4})}.
\end{alignat*}
By Theorem \ref{code_dim} we get that $\mathcal{C}$ is a $D_{20}$-code of dimension $23$, and applying \cite[Theorem 5]{Vedenev2021} we deduce
\begin{alignat*}{5}
\mathcal{C}^{\perp}\cong(\pmb{0}\oplus \pmb{0})\oplus(\mathbb{F}_3 \oplus \pmb{0})&{}\oplus  \left\langle \begin{pmatrix}
0&0\\0&0 \end{pmatrix} \right\rangle_{M_2(\mathbb{F}_{3})}&&{}\oplus  \left\langle \begin{pmatrix}
1&0\\0&0 \end{pmatrix} \right\rangle_{M_2(\mathbb{F}_{3}(\alpha_4+\alpha_4^{-1}))} \\
&{}\oplus  \left\langle \begin{pmatrix}
0&1\\0&0 \end{pmatrix} \right\rangle_{M_2(\mathbb{F}_{3}(\alpha_5+\alpha_5^{-1}))}&&{}\oplus \left\langle \begin{pmatrix}
1&0\\0&0 \end{pmatrix} \right\rangle_{M_2(\mathbb{F}_{3^4})}.
\end{alignat*}
As in the previous example,  any $D_{20}$-code containing $\mathcal{C}^{\perp}$ allows us to apply the CSS construction in order to obtain a dihedral quantum code.  More concretely,  we choose the following $D_{20}$-code $\mathcal{D}$ such that $\mathcal{C}^{\perp}\subseteq \mathcal{D}$:
 \begin{alignat*}{5}
\mathcal{D}\cong(\mathbb{F}_3\oplus \mathbb{F}_3)\oplus(\mathbb{F}_3 \oplus \pmb{0})&{}\oplus  \left\langle \begin{pmatrix}
0&1\\0&0 \end{pmatrix} \right\rangle_{M_2(\mathbb{F}_{3})}&&{}\oplus  \left\langle \begin{pmatrix}
1&0\\0&1 \end{pmatrix} \right\rangle_{M_2(\mathbb{F}_{3}(\alpha_4+\alpha_4^{-1}))}\\
&{}\oplus  \left\langle \begin{pmatrix}
0&1\\0&0 \end{pmatrix} \right\rangle_{M_2(\mathbb{F}_{3}(\alpha_5+\alpha_5^{-1}))}&&{}\oplus \left\langle \begin{pmatrix}
1&0\\0&1 \end{pmatrix} \right\rangle_{M_2(\mathbb{F}_{3^4})}.
\end{alignat*}
It holds that $\mathcal{D}$ is a $D_{20}$-code of dimension $33$. Consequently,  by  \cref{quantum_CSS}, and through computations with \textsc{Magma} \cite{Magma} to obtain the distance, we conclude that ${\cal Q}_2=\operatorname{CSS}({\cal D}, {\cal C})$ is a quantum CSS $D_{20}$-code with parameters $[\![40,16,4]\!]_{3}$. 
\end{example}

Unfortunately, none of our numerical experiments achieved the parameters of the best codes listed at \cite{Grasslcodetables}.
Nevertheless, the literature of quantum codes suggests that this usually occurs with Euclidean dualities, while Hermitian dualities (with the corresponding quantum construction) may provide quantum codes with better parameters, possibly due to the larger field size. In fact, in our recent paper \cite{Nostre2}, among other results, we compute the Hermitian duality of $D_n$-codes over $\F_{q^2}$ and we obtain quantum dihedral codes, based on this methodical strategy, that achieve the best-known minimum distance for their length and dimension.

%
%




\end{document}